\pdfoutput=1

\documentclass[a4paper,USenglish,numberwithinsect,thm-restate,cleveref]{lipics-v2019}

\usepackage{etoolbox}
\usepackage{stmaryrd}
\usepackage{amsmath}
\usepackage{amssymb}
\usepackage{mathpartir}
\usepackage{xspace}
\usepackage{xifthen}        % for conditional commands
\usepackage{mathtools}
\usepackage{skull}

\usepackage[disable,textwidth=3cm]{todonotes}
\numberwithin{equation}{section}

\newboolean{if-as-case}
\setboolean{if-as-case}{true}

\newboolean{showcomments}
\setboolean{showcomments}{true}

\ifthenelse{\boolean{showcomments}}
  {\newcommand{\nb}[2]{
    \fbox{\bfseries\sffamily\scriptsize#1}
    {\sf\small$\blacktriangleright$\textit{#2}$\blacktriangleleft$}
   }
  }
  {\newcommand{\nb}[2]{}
  }

\newboolean{longversion}
\setboolean{longversion}{true}
\newcommand{\Luca}[2][]{\todo[color=orange!66,size=\tiny,#1]{\sffamily\textbf{Luca}: #2}\xspace}

\newcommand{\eMr}[2]{\REVISION[eM]{#1}}
\newcommand{\cutout}[1]{\marginpar{\tiny \hl{CUT}}}

\newcommand{\REVISION}[2][]{#2}

\definecolor{myblue}{rgb}{0,0,0.6}
\definecolor{myred}{rgb}{0.4,0,0}
\definecolor{mygreen}{rgb}{0,0.25,0}

\newcommand{\quo}[1]{\lq\lq {#1}\rq\rq}

\newcommand{\ttbar}{\texttt{\upshape|}}

\newcommand{\mkkeyword}[1]{\mathtt{\color{myblue}#1}}
\newcommand{\mktype}[1]{\mathsf{\color{myred}#1}}

\newcommand{\set}[2][]{\{#2\}\ifblank{#1}{}{_{#1}}}

\newcommand{\seqof}[1]{\overline{#1}}

\newcommand{\rulename}[1]{\textsc{\footnotesize#1}}

\newcommand{\eg}{\emph{e.g.}\xspace}
\newcommand{\ie}{\emph{i.e.}\xspace}
\newcommand{\cf}{\emph{cf.}\xspace}

\newcommand{\proofcase}[1]{\medskip\noindent\fbox{#1}}
\newcommand{\proofrule}[1]{\proofcase{\rulename{#1}}}

\newcommand{\parens}[1]{(#1)}

\newcommand{\angles}[1]{\langle#1\rangle}

\newcommand{\eoe}{\hfill$\blacksquare$}

\newcommand{\p}{p}
\newcommand{\q}{q}
\renewcommand{\r}{r}

\renewcommand{\t}{t}
\newcommand{\s}{s}
\newcommand{\tsession}[1]{\angles{#1}}
\newcommand{\tunit}{\mktype{unit}}
\newcommand{\tint}{\mktype{int}}
\newcommand{\tjob}{\tint}

\newcommand{\T}{T}
\renewcommand{\S}{S}
\newcommand{\R}{R}

\newcommand{\tdone}{{\bullet}}
\newcommand{\tend}{{\circ}}

\newcommand{\cin}[2][.]{{?}#2#1}
\newcommand{\cout}[2][.]{{!}#2#1}

\ifthenelse{\boolean{if-as-case}}{
  
}{
  
}
\newcommand{\xpoint}[3]{#1 \mathrel{#2} #3}
\newcommand{\tbranch}[3][p]{\xpoint{#2}{\prescript{}{#1}{\&}}{#3}}
\newcommand{\tchoice}[3][p]{\xpoint{#2}{\prescript{}{#1}{\oplus}}{#3}}

\newcommand{\scong}{\equiv}
\newcommand{\spcong}{\preccurlyeq}

\newcommand{\un}[1]{\mathsf{un}\parens{#1}}
\newcommand{\safe}[1]{\mathsf{safe}\parens{#1}}
\newcommand{\bal}[1]{\mathsf{bal}\parens{#1}}
\newcommand{\eqdef}{\mathrel{\smash{\stackrel{\textsf{\upshape\tiny def}}=}}}

\newcommand{\success}[2]{\,\mathclose{\uparrow}^{#1}_{#2}}
\newcommand{\lsuccess}[2]{\,\mathclose{\Uparrow}^{#1}_{#2}}
\newcommand{\tred}[1][]{\leadsto\ifblank{#1}{}{_{#1}}}
\newcommand{\terminated}[1]{#1\,\mathclose\downarrow}
\newcommand{\terminates}[2]{#1\,\mathclose\Downarrow\ifblank{#1}{}{_{#2}}}
\newcommand{\probeq}{\simeq}

\newcommand{\fn}[1]{\mathsf{fn}(#1)}
\newcommand{\bn}[1]{\mathsf{bn}(#1)}
\newcommand{\dom}[1]{\mathsf{dom}(#1)}

\newcommand{\pr}[1]{\llbracket#1\rrbracket}

\newcommand{\dualof}[1]{\overline{#1}}

\newcommand{\csum}[3][p]{#2 \mathrel{\prescript{}{#1}{\boxplus}} #3}

\newcommand{\trees}[1]{\mathcal{T}\parens{#1}}

\newcommand{\nameset}{\mathcal{N}}
\newcommand{\x}{x}
\newcommand{\y}{y}
\newcommand{\z}{z}

\newcommand{\A}{A}
\newcommand{\B}{B}

\renewcommand{\P}{P}
\newcommand{\Q}{Q}

\newcommand{\Pnf}{\P_{\mathit{nf}}}
\newcommand{\Qnf}{\Q_{\mathit{nf}}}

\newcommand{\pidle}{\mkkeyword{idle}}
\newcommand{\pdone}[1]{\mkkeyword{done}\,#1}
\newcommand{\pin}[3][.]{#2{?}\parens{#3}#1}
\newcommand{\pout}[3][.]{#2{!}{#3}#1}

\newcommand{\pleft}[2][.]{\mkkeyword{inl}\,#2#1}
\newcommand{\pright}[2][.]{\mkkeyword{inr}\,#2#1}
\newcommand{\pbranch}[3]{\mkkeyword{case}\,#1\,[#2,#3]}

\ifthenelse{\boolean{if-as-case}}{
  
}{
  
}

\newcommand{\pnew}[2][]{(#2)}

\newcommand{\ppar}{\mathbin{\ttbar}}
\newcommand{\pvar}[2]{#1\ifthenelse{\isempty{#2}}{}{\angles{#2}}}
\newcommand{\pdef}[3]{#1\ifblank{#2}{}{\parens{#2}} := #3}

\newcommand{\C}{\mathcal{C}}
\newcommand{\D}{\mathcal{D}}
\newcommand{\hole}{[~]}

  \newcommand{\job}{\textit{job}} 
\newcommand{\red}{\rightarrow}
\newcommand{\wred}{\Rightarrow}

\newcommand{\lred}{\rightarrow}
\newcommand{\nred}{\arrownot\lred}

\newcommand{\EmptyContext}{\emptyset}

\newcommand{\Context}{\ContextA}
\newcommand{\ContextA}{\Gamma}
\newcommand{\ContextB}{\Delta}

\newcommand{\HyperContext}{\mathcal{H}}

\newcommand{\wtp}[2]{#1 \vdash #2}
\newcommand{\wtc}[2]{#1 \vdash #2}

\newcommand{\cpar}{\fatsemi}

\def\introrule{{\cal I}}\def\elimrule{{\cal E}}%%
\newdimen\proofrulebreadth \proofrulebreadth=.05em
\newdimen\proofdotseparation \proofdotseparation=1.25ex
\newdimen\proofrulebaseline \proofrulebaseline=2ex
\newcount\proofdotnumber \proofdotnumber=3
\let\then\relax
\def\hfi{\hskip0pt plus.0001fil}
\mathchardef\squigto="3A3B
\newif\ifinsideprooftree\insideprooftreefalse
\newif\ifonleftofproofrule\onleftofproofrulefalse
\newif\ifproofdots\proofdotsfalse
\newif\ifdoubleproof\doubleprooffalse
\let\wereinproofbit\relax
\newdimen\shortenproofleft
\newdimen\shortenproofright
\newdimen\proofbelowshift
\newbox\proofabove
\newbox\proofbelow
\newbox\proofrulename
\def\shiftproofbelow{\let\next\relax\afterassignment\setshiftproofbelow\dimen0 }
\def\shiftproofbelowneg{\def\next{\multiply\dimen0 by-1 }%
\afterassignment\setshiftproofbelow\dimen0 }
\def\setshiftproofbelow{\next\proofbelowshift=\dimen0 }
\def\setproofrulebreadth{\proofrulebreadth}

\def\prooftree{% NESTED ZERO (\ifonleftofproofrule)
\ifnum  \lastpenalty=1
\then   \unpenalty
\else   \onleftofproofrulefalse
\fi
\ifonleftofproofrule
\else   \ifinsideprooftree
        \then   \hskip.5em plus1fil
        \fi
\fi
\bgroup% NESTED ONE (\proofbelow, \proofrulename, \proofabove,
\setbox\proofbelow=\hbox{}\setbox\proofrulename=\hbox{}%
\let\justifies\proofover\let\leadsto\proofoverdots\let\Justifies\proofoverdbl
\let\using\proofusing\let\[\prooftree
\ifinsideprooftree\let\]\endprooftree\fi
\proofdotsfalse\doubleprooffalse
\let\thickness\setproofrulebreadth
\let\shiftright\shiftproofbelow \let\shift\shiftproofbelow
\let\shiftleft\shiftproofbelowneg
\let\ifwasinsideprooftree\ifinsideprooftree
\insideprooftreetrue
\setbox\proofabove=\hbox\bgroup$\displaystyle % NESTED TWO
\let\wereinproofbit\prooftree
\shortenproofleft=0pt \shortenproofright=0pt \proofbelowshift=0pt
\onleftofproofruletrue\penalty1
}

\def\eproofbit{% NESTED TWO
\ifx    \wereinproofbit\prooftree
\then   \ifcase \lastpenalty
        \then   \shortenproofright=0pt  % 0: some other object, no indentation
        \or     \unpenalty\hfil         % 1: empty hypotheses, just glue
        \or     \unpenalty\unskip       % 2: just had a tree, remove glue
        \else   \shortenproofright=0pt  % eh?
        \fi
\fi
\global\dimen0=\shortenproofleft
\global\dimen1=\shortenproofright
\global\dimen2=\proofrulebreadth
\global\dimen3=\proofbelowshift
\global\dimen4=\proofdotseparation
\global\count255=\proofdotnumber
$\egroup  % NESTED ONE
\shortenproofleft=\dimen0
\shortenproofright=\dimen1
\proofrulebreadth=\dimen2
\proofbelowshift=\dimen3
\proofdotseparation=\dimen4
\proofdotnumber=\count255
}

\def\proofover{% NESTED TWO
\eproofbit % NESTED ONE
\setbox\proofbelow=\hbox\bgroup % NESTED TWO
\let\wereinproofbit\proofover
$\displaystyle
}%
\def\proofoverdbl{% NESTED TWO
\eproofbit % NESTED ONE
\doubleprooftrue
\setbox\proofbelow=\hbox\bgroup % NESTED TWO
\let\wereinproofbit\proofoverdbl
$\displaystyle
}%
\def\proofoverdots{% NESTED TWO
\eproofbit % NESTED ONE
\proofdotstrue
\setbox\proofbelow=\hbox\bgroup % NESTED TWO
\let\wereinproofbit\proofoverdots
$\displaystyle
}%
\def\proofusing{% NESTED TWO
\eproofbit % NESTED ONE
\setbox\proofrulename=\hbox\bgroup % NESTED TWO
\let\wereinproofbit\proofusing
\kern0.3em$
}

\def\endprooftree{% NESTED TWO
\eproofbit % NESTED ONE
  \dimen5 =0pt% spread of hypotheses
\dimen0=\wd\proofabove \advance\dimen0-\shortenproofleft
\advance\dimen0-\shortenproofright
\dimen1=.5\dimen0 \advance\dimen1-.5\wd\proofbelow
\dimen4=\dimen1
\advance\dimen1\proofbelowshift \advance\dimen4-\proofbelowshift
\ifdim  \dimen1<0pt
\then   \advance\shortenproofleft\dimen1
        \advance\dimen0-\dimen1
        \dimen1=0pt
        \ifdim  \shortenproofleft<0pt
        \then   \setbox\proofabove=\hbox{%
                        \kern-\shortenproofleft\unhbox\proofabove}%
                \shortenproofleft=0pt
        \fi
\fi
\ifdim  \dimen4<0pt
\then   \advance\shortenproofright\dimen4
        \advance\dimen0-\dimen4
        \dimen4=0pt
\fi
\ifdim  \shortenproofright<\wd\proofrulename
\then   \shortenproofright=\wd\proofrulename
\fi
\dimen2=\shortenproofleft \advance\dimen2 by\dimen1
\dimen3=\shortenproofright\advance\dimen3 by\dimen4
\ifproofdots
\then
        \dimen6=\shortenproofleft \advance\dimen6 .5\dimen0
        \setbox1=\vbox to\proofdotseparation{\vss\hbox{$\cdot$}\vss}%
        \setbox0=\hbox{%
                \advance\dimen6-.5\wd1
                \kern\dimen6
                $\vcenter to\proofdotnumber\proofdotseparation
                        {\leaders\box1\vfill}$%
                \unhbox\proofrulename}%
\else   \dimen6=\fontdimen22\the\textfont2 % height of maths axis
        \dimen7=\dimen6
        \advance\dimen6by.5\proofrulebreadth
        \advance\dimen7by-.5\proofrulebreadth
        \setbox0=\hbox{%
                \kern\shortenproofleft
                \ifdoubleproof
                \then   \hbox to\dimen0{%
                        $\mathsurround0pt\mathord=\mkern-6mu%
                        \cleaders\hbox{$\mkern-2mu=\mkern-2mu$}\hfill
                        \mkern-6mu\mathord=$}%
                \else   \vrule height\dimen6 depth-\dimen7 width\dimen0
                \fi
                \unhbox\proofrulename}%
        \ht0=\dimen6 \dp0=-\dimen7
\fi
\let\doll\relax
\ifwasinsideprooftree
\then   \let\VBOX\vbox
\else   \ifmmode\else$\let\doll=$\fi
        \let\VBOX\vcenter
\fi
\VBOX   {\baselineskip\proofrulebaseline \lineskip.2ex
        \expandafter\lineskiplimit\ifproofdots0ex\else-0.6ex\fi
        \hbox   spread\dimen5   {\hfi\unhbox\proofabove\hfi}%
        \hbox{\box0}%
        \hbox   {\kern\dimen2 \box\proofbelow}}\doll%
\global\dimen2=\dimen2
\global\dimen3=\dimen3
\egroup % NESTED ZERO
\ifonleftofproofrule
\then   \shortenproofleft=\dimen2
\fi
\shortenproofright=\dimen3
\onleftofproofrulefalse
\ifinsideprooftree
\then   \hskip.5em plus 1fil \penalty2
\fi
}

\title{Probabilistic Analysis of Binary Sessions}

\author{Omar Inverso}{Gran Sasso Science Institute,
  Italy}{}{https://orcid.org/0000-0002-9348-1979}{}

\author{Hern\'an Melgratti}{ICC -- Universidad de Buenos Aires --
  Conicet,
  Argentina}{}{https://orcid.org/0000-0003-0760-0618}{}

\author{Luca Padovani}{Universit\`a di Torino,
  Italy}{}{https://orcid.org/0000-0001-9097-1297}{}

\author{Catia Trubiani}{Gran Sasso Science Institute,
  Italy}{}{https://orcid.org/0000-0002-7675-6942}{}

\author{Emilio Tuosto}{Gran Sasso Science Institute,
  Italy}{}{https://orcid.org/0000-0002-7032-3281}{}

\authorrunning{O. Inverso and H. Melgratti and L. Padovani and
  C. Trubiani and E. Tuosto} \Copyright{Omar Inverso and Hern\'an
  Melgratti and Luca Padovani and Catia Trubiani and Emilio Tuosto}



\acknowledgements{\REVISION{The authors are grateful to the
    anonymous reviewers for their detailed feedback.  }}

\funding{%
  Omar Inverso has been partially supported by MIUR project PRIN
  2017FTXR7S \emph{IT MATTERS} (Methods and Tools for Trustworthy
  Smart Systems).
  Catia Trubiani has been partially supported by MIUR project PRIN
  2017TWRCNB \emph{SEDUCE} (Designing Spatially Distributed
  Cyber-Physical Systems under Uncertainty).
  Hern\'an Melgratti, Luca Padovani and Emilio Tuosto have been
  partially supported by EU H2020 RISE programme under the Marie
  Sk{\l}odowska-Curie grant agreement No 778233.  
  Hern\'an Melgratti has been partially supported by UBACyT projects 
  20020170100544BA and 20020170100086BA and 
  PIP project 11220130100148CO.
  }

\nolinenumbers %uncomment to disable line numbering

\EventEditors{Igor Konnov and Laura Kov\'{a}cs}
\EventNoEds{2}
\EventLongTitle{31st International Conference on Concurrency Theory (CONCUR 2020)}
\EventShortTitle{CONCUR 2020}
\EventAcronym{CONCUR}
\EventYear{2020}
\EventDate{September 1--4, 2020}
\EventLocation{Vienna, Austria}
\EventLogo{}
\SeriesVolume{2017}
\ArticleNo{36}

\begin{document}

\maketitle

\begin{abstract}
  We study a probabilistic variant of binary session types
  that relate to a class of Finite-State Markov Chains. The
  probability annotations in session types enable the reasoning on
  the probability that a session terminates successfully, for some
  user-definable notion of successful termination.
  We develop a type system for a simple session calculus featuring
  probabilistic choices and show that the success probability of
  well-typed processes agrees with that of the sessions they use.
  To this aim, the type system needs to track the propagation of
  probabilistic choices across different sessions.
\end{abstract}

\section{Introduction}
\label{sec:introduction}

Session types~\cite{Honda93,HuttelEtAl16} have consolidated as a
formalism for the modular analysis of complex systems of
communicating processes. A \emph{session} is a private channel
connecting two (sometimes more) processes, each owning one
\emph{endpoint} of the session and using the endpoint according to a
specification -- the \emph{session type} -- that constrains the
sequence of messages that can be sent and received through that
endpoint. As an example, the session type
\begin{equation}
  \label{eq:bare.auction}
  \cout\tint\parens{
    \tbranch[]\tend{
      \cin\tint\parens{
        \tchoice[]\tend\T
      }
    }
  }
\end{equation}
could describe (part of) an auction protocol as seen from the
viewpoint of a buyer process, which sends a bid ($\cout[]\tint$) and
waits for a decision from the auctioneer. The protocol proceeds in
two different ways, as specified by the two sides of the branching
operator $\tbranch[]{}{}$.  The auctioneer may declare that the item
is sold, in which case the session terminates immediately ($\tend$),
or it may inform the buyer of a different (higher) bid
($\cin[]\tint$). At that point the buyer may choose
($\tchoice[]{}{}$) to quit the auction or to restart the same
protocol, here denoted by $T$, with another bid.

Most session type theories are aimed at enforcing qualitative
properties of a system, such as type safety, protocol compliance,
deadlock and livelock freedom, and so on~\cite{HuttelEtAl16}. In
these theories, branches ($\tbranch[]{}{}$) and choices
($\tchoice[]{}{}$) are given a non-deterministic interpretation
since all that matters is understanding whether the system ``behaves
well'' \REVISION{no matter how it evolves}.
In this work, we propose a session type system for \emph{\REVISION{a
    particular} quantitative analysis} of session-based networks of
communicating processes. More specifically, we shift from a
non-deterministic to a \emph{probabilistic} interpretation of
branches and choices in session types and study a type system aimed
at determining the probability with which a particular session
terminates \emph{successfully}. Since there is no universal
interpretation of ``successful termination'', we differentiate
successful from unsuccessful termination of a session by means of a
dedicated type constructor.
For example, in our type system we can refine
\eqref{eq:bare.auction} as
\begin{equation}
  \label{eq:refined.auction}
  \cout\tint\parens{
    \tbranch[\p]\tdone{
      \cin\tint\parens{
        \tchoice[\q]\tend\T
      }
    }
  }
\end{equation}
where the session type $\tdone$ indicates successful termination and
branches and choices are annotated with probabilities $p$ and
$q$. In particular, the auctioneer declares the item sold with
probability $p$ and answers with a counteroffer with probability
$1-p$, whereas the buyer decides to quit the auction with
probability $q$ and to bid again with probability $1-q$.

From an abstract description such as \eqref{eq:refined.auction}, we
can easily compute the probability that the interaction ends up in a
particular state (\eg, the probability with which the buyer wins the
auction).
However, \eqref{eq:refined.auction} is ``just'' the type of one
endpoint of a single session in a system, while the system itself
could be much more complex: there could be many different processes
involved, each making probabilistic choices affecting the behavior
of faraway processes that directly or indirectly receive information
about such choices through messages exchanged in sessions. Also, new
processes and sessions could be created and the network topology
could evolve dynamically as the system runs. How do we know that
\eqref{eq:refined.auction} is a faithful abstraction of our system?
How do we know that the probability annotations we see in
\eqref{eq:refined.auction} correspond to the actual probabilities
that the system evolves in a certain way?  Here is where our type
system comes into play: by certifying that a system of processes is
well typed with respect to a given set of session types with
probability annotations, we support the computation of the
probability that the system evolves in certain way statically -- \ie,
before the system runs -- and solely looking at the session types we
are interested in as opposed to the system itself.

\subparagraph*{Summary of contributions and structure of the paper.}
We define a session calculus in which processes may perform
probabilistic choices (\cref{sec:model}). We study a variant of
session types based on a probabilistic interpretation of branches
and choices so that session types correspond to a particular class
of Discrete-Time Markov Chains (\cref{sec:types}). We provide
syntax-directed typing rules for relating processes and session
types (\cref{sec:rules}). Well-typed processes are shown to behave
probabilistically as specified by the corresponding session
types. We are able to trace this correspondence not just for finite
processes (\cref{thm:soundness}) but also for processes engaged in
potentially infinite interactions (\cref{cor:relative.success}).
We discuss related work in \cref{sec:related} and ideas for further
developments in \cref{sec:conclusion}.
Example details and proofs of all the results are relegated to the appendices.

\section{A Probabilistic Session Calculus}
\label{sec:model}

\begin{table}
  \framebox[\textwidth]{
    \begin{math}
      \displaystyle
      \begin{array}[t]{@{}rr@{~}c@{~}ll@{}}
        \textbf{Domains}
        & \p, \q, \r & \in & [0, 1] & \text{probability}
        \\
        & x, y, z & \in & \nameset & \text{name}
        \\\\
        \textbf{Processes}
        & \P, \Q & ::= & \pidle & \text{inaction} \\
        & & | & \pdone\x & \text{success} \\
        & & | & \pin\x\y\P & \text{message input} \\
        & & | & \pout\x\y\P & \text{message output} \\
      \end{array}
      \quad
      \begin{array}[t]{@{}r@{~}c@{~}ll@{}}
        & | & \pbranch\x\P\Q & \text{branch} \\
        & | & \pleft\x\P & \text{left selection} \\
        & | & \pright\x\P & \text{right selection} \\
        & | & \P \ppar \Q & \text{parallel composition} \\
        & | & \pnew\x\P & \text{session restriction} \\
        & | & \csum[\p]\P\Q & \text{probabilistic choice} \\
        & | & \pvar\A{\seqof\x} & \text{process invocation} \\
      \end{array}
    \end{math}
  }
  \caption{\label{tab:model} Syntax of processes.}
\end{table}

We let $\p$, $\q$ and $\r$ range over \emph{probabilities}, namely
real numbers in the range $[0,1]$. We let $\x$, $\y$ and $z$ range
over an infinite set $\nameset$ of \emph{channel names}. We write
$\seqof\x$ for finite sequences of names and other entities.
Processes, ranged over by $\P$, $\Q$ and $R$, are defined by the
grammar in~\cref{tab:model}.
We have two distinct terms, $\pidle$ and $\pdone\x$, for modeling
inactive processes. We use $\pidle$ to denote plain termination and
$\pdone\x$ to denote successful termination of session $x$. This
way, we are able to relate the success rate resulting from processes
to that inferrable from session types (\cref{thm:soundness}).
The terms $\pin\x\y\P$ and $\pout\x\y\P$ denote a process that
respectively performs an input and an output of a message $\y$ on
session $\x$ and then continues as $\P$. For simplicity, in the
model we only consider messages that are themselves (session)
channels, while in some examples we will also use more elaborate
message types.
The term $\pbranch\x\P\Q$ represents a process that waits for a
selection (either ``left'' or ``right'') on session $x$ and
continues as either $P$ or $Q$ accordingly.
The terms $\pleft\x\P$ and $\pright\x\P$ represent processes that
perform a selection (respectively ``left'' and ``right'') on session
$x$ and continue as $P$.
Parallel composition $\P \ppar \Q$, channel restriction $\pnew\x\P$
and process invocation $\pvar\A{\seqof\x}$ are standard. We assume
that for every process variable $\A$ there is an equation
$\pdef\A{\seqof\x}\P$ defining it.
Finally, the term $\csum\P\Q$ represents a process that has
performed a probabilistic choice and that behaves as $P$ with
probability $p$ and as $Q$ with probability $1-p$.

The notions of free and bound names are standard. In the following,
we write $\fn\P$ and $\bn\P$ for the set of free and bound names of
$\P$, respectively.
For the sake of readability, we occasionally omit $\pidle$ terms and
we assume that input/output prefixes and selections bind more
tightly than choices and parallel compositions. So for example,
$\csum{\pleft\x\pdone\y}{\pright[]\x}$ is to be read
$\csum{(\pleft\x\pdone\y)}{(\pright\x\pidle)}$.

\begin{table}
  \framebox[\textwidth]{
    \begin{math}
      \displaystyle
      \begin{array}{@{}c@{}}
        \multicolumn{1}{@{}l@{}}{\textbf{Structural pre-congruence}\hfill\fbox{$\mathstrut\P \spcong \Q$}}
        \\\\
        \inferrule[\rulename{s-no-choice}]{}{
          \csum[1]\P\Q \spcong \P
        }
        \qquad
        \inferrule[\rulename{s-choice-idem}]{}{
          \csum\P\P \scong \P
        }
        \qquad
        \inferrule[\rulename{s-choice-comm}]{}{
          \csum[\p]\P\Q \scong \csum[1-\p]\Q\P
        }
        \qquad
        \inferrule[\rulename{s-new-comm}]{}{
          \pnew\x\pnew[q]\y\P \scong \pnew[q]\y\pnew\x\P
        }
        \\\\
        \inferrule[\rulename{s-par-comm}]{}{
          \P \ppar \Q \scong \Q \ppar \P
        }
        \qquad
        \inferrule[\rulename{s-par-choice}]{}{
          (\csum[\p]\P\Q) \ppar \R \spcong \csum[\p]{(\P \ppar \R)}{(\Q \ppar \R)}
        }
        \qquad
        \inferrule[\rulename{s-par-new}]{
          x\not\in\fn\Q
        }{
          \pnew\x\P \ppar \Q \scong \pnew\x(\P \ppar \Q)
        }
        \\\\
        \inferrule[\rulename{s-choice-assoc}]{
          pq < 1
        }{
          \csum[\p]{(\csum[\q]\P\Q)}\R \scong
          \csum[\p\q]\P{(\csum[\frac{\p-\p\q}{1-\p\q}]\Q\R)}
        }
        \qquad
        \inferrule[\rulename{s-par-assoc}]{
          \fn\Q \cap \fn\R \ne \emptyset
        }{
          (\P \ppar \Q) \ppar \R \spcong \P \ppar (\Q \ppar \R)
        }
        \\\\
        \multicolumn{1}{@{}l@{}}{\textbf{Reduction}\hfill\fbox{$\mathstrut\P \lred \Q$}}
        \\\\
        \inferrule[\rulename{r-com}]{}{
          \pout\x\y\P \ppar \pin\x\y\Q \lred \P \ppar \Q
        }
        \qquad
        \inferrule[\rulename{r-left}]{}{
          \pleft\x\P \ppar \pbranch\x\Q\R \lred \P \ppar \Q
        }
        \qquad
        \inferrule[\rulename{r-var}]{
          \pdef\A{\seqof\x}\P
        }{
          \pvar\A{\seqof\x} \lred \P
        }
        \\\\
        \inferrule[\rulename{r-par}]{
          \P \lred \Q
        }{
          \P \ppar \R \lred \Q \ppar \R
        }
        \qquad
        \inferrule[\rulename{r-new}]{
          \P \lred \Q
        }{
          \pnew\x\P \lred \pnew\x\Q
        }
        \qquad
        \inferrule[\rulename{r-choice}]{
          \P \lred \Q
        }{
          \csum[\p]\P\R \lred \csum[\p]\Q\R
        }
        \qquad
        \inferrule[\rulename{r-struct}]{
          \P \spcong \R \lred \R' \spcong \Q
        }{
          \P \lred \Q
        }
      \end{array}
    \end{math}
  }
  \caption{\label{tab:semantics} Structural pre-congruence and reduction of processes.}
\end{table}

The operational semantics of processes is given by a structural
precongruence relation $\spcong$ and a reduction relation $\red$,
which are defined by the axioms and rules in \cref{tab:semantics}
where we abbreviate with $\P \scong \Q$ the two relations
$P \spcong Q$ and $Q \spcong P$. We use a pre-congruence instead of
a symmetric relation because careless rewriting of processes may
compromise their well typing. Nonetheless, the use of a
pre-congruence does not affect the ability of processes to reduce
(\cf \cref{thm:df}) and most relations are symmetric anyway.
We now describe the structural pre-congruence and reduction,
focusing on the former relation since it is the only one that deals
with probabilistic choices.

The relations described by \rulename{s-par-comm},
\rulename{s-new-comm} and \rulename{s-par-new} are standard and need
no commentary.
Axiom \rulename{s-choice-comm} allows us to commute a probabilistic
choice. The probability needs to be suitably adjusted so as to
preserve the semantics of the process.
Axiom \rulename{s-no-choice} turns a probabilistic choice into a
deterministic one when the probability is trivial. \REVISION{This
  axiom is the main motivation for adopting a pre-congruence rather
  than a symmetric relation. Indeed, while the symmetric relation
  $P \spcong \csum[1]\P\Q$ makes sense operationally, it violates
  typing in general for the process $Q$ can be arbitrary.}  On the
contrary, knowing that $\csum[1]\P\Q$ is well typed allows us to
easily derive that $\P$ alone is also well typed.
Axiom \rulename{s-choice-idem} states that the probabilistic choice
is idempotent, namely that a probabilistic choice between equal
behaviors is not really a choice.
Rule \rulename{s-choice-assoc} expresses the standard associativity
property for probabilistic choices, which requires a normalization
of the involved probabilities. Note that this rule is applicable
only when $pq<1$, or else the rightmost probability in the
conclusion would be undefined. When $pq = 1$, the process can be
simplified using \rulename{s-no-choice}.
Rule \rulename{s-par-assoc} expresses the associativity property for
the parallel composition.  The side condition, requiring the middle
and rightmost processes to be connected by one shared name, is
needed by the type system (\cf \cref{sec:rules}). The reader might
\REVISION{be worried by the} side conditions imposed on the
associativity rule, since \REVISION{they} are limiting the ability
to rewrite processes to an extent which could prevent processes to
be placed next to each other and reduce according to the reduction
relation.  \REVISION{It is possible to prove a} \emph{proximity
  property} (\cref{lem:proximity})
ensuring that this is not the case, namely that it is always
possible to rearrange (well-typed) processes in such a way that
processes connected by a session can communicate. The symmetric
relation $\P \ppar (Q \ppar R) \spcong (P \ppar Q) \ppar R$ when
$\fn\P \cap \fn\Q \ne \emptyset$ is derivable using
\rulename{s-par-assoc} and \rulename{s-par-comm}.
Rule \rulename{s-par-choice} distributes parallel compositions over
probabilistic choices. This rule is pivotal in our model, for two
different reasons. First, being able to distribute a process over a
probabilistic choice is essential to make sure that processes
connected by a session can be placed next to each other so that they
can reduce according to $\red$. Second, the relation is quite
challenging to handle at the typing level: when $R$ is composed in
parallel with $P$ and $Q$, it might be necessary to type $R$
differently depending on whether or not the session that connects
$\R$ with $\P$ and $\Q$ is affected by the probabilistic
choice. This is doable provided that $R$ uses the session
\emph{safely}, namely if it does not delegate the session before it
becomes aware of the probabilistic choice (\cf \cref{sec:rules}).

The reduction relation is standard. The base cases consist of the
usual rules for communication \eMr{(\rulename{r-com})}{(\rulename{t-com})}, branch selection
(\rulename{r-left} and \rulename{r-right}, the latter omitted) along
with the expansion of process variables
(\rulename{r-var}). Reduction is closed under parallel compositions
(\rulename{r-par}), restrictions (\rulename{r-new}), probabilistic
choices (\rulename{r-choice}) and structural precongruence
(\rulename{r-struct}).
Note that a probabilistic choice $\csum\P\Q$ is \emph{persistent},
in the sense that neither $P$ nor $Q$ is discarded by reduction even
though they morally represent two mutually-exclusive evolutions of
the same process. This is one of the standard approaches for
describing the semantics of probabilistic
processes~\cite{DBLP:conf/fossacs/HerescuP00,varacca2007probabilistic,leventis2019deterministic}. As
a consequence, a process like $\pdef\A{}{\csum[0.001]\pidle\A}$
diverges but terminates with probability 1. We will be able to state
interesting properties of such processes through a soundness result
that is relativized to the probability of termination.

We write $\wred$ for the reflexive, transitive closure of $\red$, we
write $\P \red {}$ if there exists $Q$ such that $P \red Q$ and
$P \nred$ if not $P \red {}$. In the above example, $A \wred \P$
implies $P \red$.

\newcommand{\Buyer}{\textit{Buyer}}
\newcommand{\Seller}{\textit{Seller}}

\begin{example}[Auction]
  \label{ex:auction}
  We end this section showing how to represent in our calculus the
  auction example informally described in \cref{sec:introduction}.
  We define two processes, a $\Buyer$ and a $\Seller$ connected by a
  session $x$:
  \[
    \begin{array}{r@{~}l}
      \pdef\Buyer\x{&
        \pout\x{\textit{bid}}
        \pbranch\x{
          \pdone\x
        }{
          \pin\x\y
          \parens{
            \csum[q]{
              \pleft[]\x
            }{
              \pright\x\pvar\Buyer\x
            }
          }
        }
      }
      \\
      \pdef\Seller\x{&
        \pin\x\z
        \parens{
          \csum[p]{
            \pleft\x
            \pdone\x
          }{
            \pright\x
            \pout\x{\textit{counteroffer}}
            \pbranch\x{
              \pidle
            }{
              \pvar\Seller\x
            }
          }
        }
      }
    \end{array}
  \]

  The buyer sends the current \emph{bid} on $x$ and waits for a
  reaction from the seller. The seller accepts the bid with
  probability $p$ and rejects it with probability $1-p$. If the
  seller accepts (by selecting the left branch of the session), the
  buyer terminates successfully. Otherwise, the seller proposes a
  counteroffer, which the buyer rejects with probability $q$ and
  accepts with probability $1 - q$. In the first case, the session
  terminates without satisfaction of the buyer. In the second case,
  the buyer starts a new negotiation.
  \eoe
\end{example}
\section{Probabilistic Session Types}
\label{sec:types}

\subparagraph*{Session types.}
Probabilistic session types describe communication protocols taking
place through session endpoints and their (finite) syntax is given
by the following grammar:
\begin{equation}
  \label{eq:session.types}
  \textbf{Session type}
  \qquad
  \T, \S ~::=~ \tend ~\mid~ \tdone ~\mid~ \cin\t\T ~\mid~ \cout\t\T ~\mid~ \tbranch\T\S ~\mid~ \tchoice\T\S
\end{equation}

The session types $\tend$ and $\tdone$ describe a session endpoint
on which no further input/output operations are possible. We use
$\tdone$ to mark those termination points of a protocol that
represent success and that we target in our probabilistic
analysis. The precise meaning of ``successful termination'' is
domain specific but also irrelevant in the technical development
that follows.
The session types $\cin\t\T$ and $\cout\t\T$ describe session
endpoints used for receiving (respectively, sending) a message of
type $\t$ and then according to $\T$. Types will be discussed
shortly.
The session types $\tbranch\T\S$ and $\tchoice\T\S$ describe a
session endpoint used for receiving (respectively, sending) a binary
choice which is ``left'' with probability $p$ and ``right'' with
probability $1-p$. The endpoint is then used according to $T$ or
$S$, respectively.
\REVISION{Note that $\tchoice{}{}$ is an internal choice -- the
  process behaving according to this type internally chooses either
  ``left'' or ``right'' -- whereas $\tbranch{}{}$ is an external
  choice -- the process behaving according to this type externally
  offers behaviors corresponding to both choices. Therefore, the
  probability annotation in $\tbranch{}{}$ is completely determined
  by the one in the corresponding internal choice and it could be
  argued that it is somewhat superfluous. Nonetheless, as we will
  see when discussing the typing rule for branch processes, having
  direct access to this annotation makes it easy to propagate the
  probability of choices across different sessions.}

We do not use any special syntax for specifying infinite session
types. Rather, we interpret the productions for $\T$ coinductively
and we call session types the possibly infinite trees generated by
the productions in \eqref{eq:session.types} that satisfy the
following conditions:
\begin{description}
\item[Regularity] We require every tree to consist of finitely many
  \emph{distinct} subtrees. This condition ensures that session
  types are finitely representable either using the so-called
  ``$\mu$ notation''~\cite{Pierce02} or as solutions of finite sets
  of equations~\cite{Courcelle83}.
\item[Reachability] We require every subtree $\T$ of a session type
  to contain a \emph{reachable leaf} labelled by $\tend$ or
  $\tdone$. This condition ensures that it is always possible to
  terminate a session regardless of how long it has been running.
\end{description}

To formalize these conditions, we define a relation $\T \tred[p] \S$
modeling the fact that (the behavior described by) $T$ may evolve
into $S$ with probability $p$ in a single step:
\[
  \begin{array}{r@{~}c@{~}l}
    \tend & \tred[1] & \tend
    \\
    \tdone & \tred[1] & \tdone
  \end{array}
  \qquad
  \begin{array}{r@{~}c@{~}l}
    \cin\t\T & \tred[1] & \T
    \\
    \cout\t\T & \tred[1] & \T
  \end{array}
  \qquad
  \begin{array}{r@{~}c@{~}l}
    \tbranch\T\S & \tred[p] & \T
    \\
    \tchoice\T\S & \tred[p] & \T
  \end{array}
  \qquad
  \begin{array}{r@{~}c@{~}l}
    \tbranch\T\S & \tred[1-p] & \S
    \\
    \tchoice\T\S & \tred[1-p] & \S
  \end{array}
\]

We also consider the relation $\tred[p]^*$, which accounts for
multiple steps in the expected way:
\[
  \inferrule{}{
    \T \tred[1]^* \T
  }
  \qquad
  \inferrule{
    \T \tred[p] \S
  }{
    \T \tred[p]^* \S
  }
  \qquad
  \inferrule{
    \T \tred[p]^* \T'
    \\
    \T' \tred[q]^* \S
  }{
    \T \tred[pq]^* \S
  }
\]

Roughly speaking, $\tred[p]^*$ is the reflexive, transitive closure
of $\tred[p]$ except that the probability annotation $p$ accounts
for the cumulative transition probability between two session types.

\begin{definition}[well-formed session type]
  Let $\trees\T \eqdef \set{ S \mid \exists p, S: T \tred[p]^* S}$.
  A (possibly infinite) tree $\T$ generated by the productions
  in~\eqref{eq:session.types} is a \emph{well-formed session type} if
  $\trees\T$ is finite and, for every $S \in \trees\T$, there exists
  $p > 0$ such that either $S \tred[p]^* \tend$ or
  $S \tred[p]^* \tdone$.
\end{definition}

\begin{example}[auction protocol, buyer side]
  \label{ex:auction.buyer}
  Even though we have not presented the typing rules for the
  calculus of \cref{sec:model}, we can speculate on the session type
  of the endpoint used \eg, by the buyer process in
  \cref{ex:auction}, which satisfies the equation
  \[
    \T = \cout\tint\parens{
      \tbranch[p]\tdone{
        \cin\tint\parens{
          \tchoice[q]\tend\T
        }
      }
    }
  \]

  In this case we have
  $\trees\T = \set{ \tend, \tdone,
    \tbranch[p]\tdone{(\cin\tint{(\tchoice[q]\tend\T))}},
    \cin\tint{(\tchoice[q]\tend\T)}, \tchoice[q]\tend\T, \T}$ and it
  is easy to see that $\T$ is well formed provided that at least one
  among $p$ and $q$ is positive.
  \eoe
\end{example}

From now on we assume that all the session types we work with are
well formed.

\subparagraph*{Success probability.} We now define the probability
that a protocol described by a session type $\T$ terminates
successfully. Intuitively, this probability is computed by
accounting for all paths in the structure of $\T$ that lead to a
leaf labelled by $\tdone$. Formally:

\begin{definition}[success probability]
  \label{def:pr}
  The \emph{success probability} of a session type $\T$, denoted by
  $\pr\T$, is determined by the following equations:
  \[
    \begin{array}{r@{~}c@{~}l}
      \pr\tend & = & 0 \\
      \pr\tdone & = & 1 \\
    \end{array}
    \qquad
    \begin{array}{r@{~}c@{~}l}
      \pr{\cin\t\T} & = & \pr\T \\
      \pr{\cout\t\T} & = & \pr\T \\
    \end{array}
    \qquad
    \begin{array}{r@{~}c@{~}l}
      \pr{\tbranch\T\S} & = & p\pr\T + (1-p)\pr\S \\
      \pr{\tchoice\T\S} & = & p\pr\T + (1-p)\pr\S \\
    \end{array}
  \]
\end{definition}

For a \emph{finite} session type $\T$, Definition~\ref{def:pr} gives
a straightforward recursive algorithm for computing $\pr\T$. When
$\T$ is infinite, however, it is less obvious that
Definition~\ref{def:pr} provides a way for determining $\pr\T$. To
address the problem in the general case we observe that, by
interpreting $\pr\T$ as a \emph{probability variable},
Definition~\ref{def:pr} allows us to derive a \emph{finite system of
  equations} relating such variables. Indeed, the right hand side of
each equation for $\pr\T$ in Definition~\ref{def:pr} is expressed in
terms of probability variables corresponding to the children nodes
in the tree of $\T$. Since $\T$ has finitely many subtrees, we end
up with finitely many equations.
Then, we observe that every session type $\T$ corresponds to a
Discrete-Time Markov Chain
(DTMC)~\cite{KemenySnell76,RuttenEtAl2004} whose state space is
$\trees\T = \set{S_1, \dots, S_n}$ and such that the probability
$p_{ij}$ of performing a transition from state $\S_i$ to state
$\S_j$ is given by
\[
  p_{ij} \eqdef
  \begin{cases}
    \p & \text{if $\S_i \tred[p] \S_j$}
    \\
    0 & \text{otherwise}
  \end{cases}
\]
\label{pg:regreach}
Regularity and reachability imply that the DTMC we obtain from any
session type $T$ is finite state and absorbing. That is, it is
always possible to reach an \emph{absorbing state} (either $\tend$
or $\tdone$) from any \emph{transient state} (any other session
type). In any finite-state, absorbing DTMC, the probability of
reaching a specific absorbing state from any transient state can be
computed by solving a particular system of equations which is
guaranteed to have a unique solution~\cite{KemenySnell76}.
Moreover, the system that we obtain for $\pr\T$ using
Definition~\ref{def:pr} is precisely the one whose solution is the
probability of reaching $\tdone$ from $\T$ (see
\cref{sec:supplement_types}).

\begin{example}
  \label{ex:auction.buyer.full}
  We compute the success probability of $T$ from
  \cref{ex:auction.buyer} where, for the sake of illustration, we
  take $p=\frac14$ and $q=\frac23$. Let
  $T_1 = \tbranch[\frac14]\tdone{T_2}$ and $T_2 = \cin\tint\T_3$ and
  $T_3 = \tchoice[\frac23]\tend\T$ be convenient names for some of
  its subtrees. Using \cref{def:pr} we obtain the system of
  equations
  \[
    \begin{array}[t]{l@{\ = \ }l}
      \pr{\T}  & \pr{T_1} \\
      \pr{T_1} & \frac14\pr\tdone+\frac34\pr{T_2} \\
    \end{array}
    \qquad
    \begin{array}[t]{l@{\ = \ }l}
      \pr\tdone & 1 \\
      \pr{T_2} & \pr{T_3} \\
    \end{array}
    \qquad
    \begin{array}[t]{l@{\ = \ }l}
      \pr{T_3} & \frac23\pr\tend+\frac13\pr{\T} \\
      \pr\tend & 0 \\
    \end{array}
  \]
  from which we compute $\pr\T = \frac13$
  (\cref{sec:supplement_types} details the corresponding DTMC).
  \eoe
\end{example}

\subparagraph*{Duality.}
We write $\dualof\T$ for the \emph{dual} of $\T$, that is the
session type obtained from $\T$ by swapping input actions with
output actions and leaving the remaining forms unchanged. Formally,
$\dualof\T$ is the session type obtained from $\T$ that satisfies
the following equations:
\[
  \begin{array}{r@{~}c@{~}l}
    \dualof\tend & = & \tend
    \\
    \dualof\tdone & = & \tdone
  \end{array}
  \qquad
  \begin{array}{r@{~}c@{~}l}
    \dualof{\cin\t\T} & = & \cout\t\dualof\T
    \\
    \dualof{\cout\t\T} & = & \cin\t\dualof\T
  \end{array}
  \qquad
  \begin{array}{r@{~}c@{~}l}
    \dualof{\tbranch\T\S} & = & \tchoice{\dualof\T}{\dualof\S}
    \\
    \dualof{\tchoice\T\S} & = & \tbranch{\dualof\T}{\dualof\S}
  \end{array}
  % \qquad
  % \dualof{\tbarrier\T} = \tbarrier{\dualof\T}
\]

It is easy to see that duality is an involution (that is,
$\dualof{\dualof\T} = \T$) and that the success probability is
unaffected by duality, that is $\pr\T = \pr{\dualof\T}$. This means
that we can compute the success probability of a session from either
of its two endpoints.

\subparagraph*{Types.}
Types describe resources used by processes and exchanged as
messages. We distinguish between session endpoints, whose type is a
session type $\T$, from sessions with success probability $p$, whose
type has the form $\tsession\p$:
\begin{equation}
  \label{eq:types}
  \textbf{Type}
  \qquad
  \t, \s ~::=~ \T ~\mid~ \tsession\p
\end{equation}

We will see in~\cref{sec:rules} that a type $\tsession\p$ results
from ``joining'' the two peer endpoints of a session having dual
sessions types $\T$ and $\dualof\T$ such that
$\p = \pr\T = \pr{\dualof\T}$.
For brevity we omit message types such as $\tunit$ and $\tint$ from
the formal development as their handling is folklore and does not
affect the presented results. We occasionally use them in the
examples though.

A key aspect of the type system is that processes may use session
endpoints differently, depending on the outcome of probabilistic
choices. Nonetheless, we need to capture the overall effect of such
different uses in a single type. For this reason, we introduce a
\emph{probabilistic type combinator} that allows us to combine types
by weighing the different ways in which a resource is used according
to a given probability.

\begin{definition}[probabilistic type combination]
  \label{def:ccomb}
  We write $\csum\t\s$ for the \emph{combination of $t$ and $s$
    weighed by $p$}, which is defined by cases on the form of $t$
  and $s$ as follows:
  \[
    \csum\t\s
    \eqdef
    \begin{cases}
      t & \text{if $t = s$}
      \\
      \tchoice[pq+(1-p)r]\T\S & \text{if $t = \tchoice[q]\T\S$ and $s = \tchoice[r]\T\S$}
      \\
      \tsession{pq+(1-p)r} & \text{if $t = \tsession\q$ and $s = \tsession\r$}
      \\
      \text{undefined} & \text{otherwise}
    \end{cases}
  \]
\end{definition}

Intuitively, $\csum\t\s$ describes a resource that is used according
to $t$ with probability $p$ and according to $s$ with probability
$1-p$. The combination of $\t$ and $\s$ is only defined when $\t$
and $\s$ have ``compatible shapes'', the trivial case being when
they are the same type. The interesting cases are when $t$ and $s$
describe a choice (a point of the protocol where one process
performs a selection) and when $t$ and $s$ describe a session as a
whole. In both cases, the success probability of the choice
(respectively, of the session) is weighed by $p$.
As an example, consider a session endpoint that is used according to
$\tchoice[1]\T\S$ with probability $p$ and according to
$\tchoice[0]\T\S$ with probability $1-p$. In the first case, we are
certain that the session endpoint is used for selecting ``left'' and
then according to $T$. In the second case, we are certain that the
session endpoint is used for selecting ``right'' and then according
to $S$. Overall, the session endpoint is used according to the type
$\tchoice[p]\T\S$.

The combination
$\csum{\tsession\q}{\tsession\r} = \tsession{pq+(1-p)r}$ captures
the fact that the success probability of a whole session that is
carried out in two different ways having success probabilities
respectively $q$ and $r$ is the convex sum of $q$ and $r$ weighed by
$p$. The success probability with which we annotate this type allows
us to state the soundness properties of the type system, by relating
the success probabilities in session types with those of a process
that behaves according to those session types.
Speaking of success probability, a fundamental property that is used
extensively in the soundness proofs is the following one. Any
conceivable generalization of \cref{def:ccomb} must guarantee this
property for the type system to be sound.

\begin{restatable}{proposition}{restateccomb}
  \label{prop:ccomb}
  $\pr{\csum{T_1}{T_2}} = p\pr{T_1} + (1-p)\pr{T_2}$.
\end{restatable}

\cref{def:ccomb} is quite conservative in that, except for top-level
choices, any other session type can only be combined with itself. It
is conceivable to generalize $\csum{}{}$ to permit the combination
of ``deep choices'' found after a common prefix. For example, we
could have
$\csum{\cout\tint(\tchoice[1]\T\S)}{\cout\tint(\tchoice[0]\T\S)} =
\cout\tint(\tchoice[p]\T\S)$. This generalization is not for free,
though. As we will see in \cref{sec:rules}, session endpoints that
are affected by a probabilistic choice must be ``handled with care''
and \cref{def:ccomb} as it stands helps ensuring that this is
actually the case. We leave the combination of ``deep choices'' to
future work.
\section{Typing Rules}
\label{sec:rules}

We use contexts for tracking the type of free variables occurring in
processes. A \emph{context} is a finite map from variables to types
written $\x_1 : \t_1, \dots, \x_n : \t_n$. We let $\ContextA$ and
$\ContextB$ range over contexts, we write $\EmptyContext$ for the
empty context, $\dom\Context$ for the domain of $\Context$ and
$\ContextA, \ContextB$ for the union of $\ContextA$ and $\ContextB$
when $\dom\ContextA \cap \dom\ContextB = \emptyset$.  We also extend
$\csum{}{}$ pointwise to contexts in the obvious
way.

Before we discuss the typing rules, we have to introduce two
predicates to single out types that have particular properties. The
class of \emph{unrestricted types}, defined next, is aimed at
describing resources that can be discarded and duplicated at will.

\begin{definition}[unrestricted type and context]
  \label{def:un}
  We say that $\t$ is \emph{unrestricted} and we write $\un\t$ if
  $\t = \tend$. We write $\un\Context$ if $\un{\Context(\x)}$ for
  all $\x\in\dom\Context$.
\end{definition}

In our case, the only unrestricted type is $\tend$, but if the type
system is extended with basic types such as $\tunit$ and $\tint$,
these would be unrestricted as well.
Next we introduce the class of \emph{safe types}, those describing
resources that can be safely sent in messages and used in process
invocations because they cannot be passively affected by a
probabilistic choice.

\begin{definition}[safe type]
  \label{def:safe}
  We write $\safe\t$ if $\t$ is \emph{not} of the form
  $\tbranch\T\S$.
\end{definition}

The ultimate motivation for the safety predicate has its roots in
the soundness proof of the type system. In a nutshell, an unsafe
session type is one whose dual admits a non-trivial probabilistic
combination (\cref{def:ccomb}) and therefore that may change
unpredictably, from the standpoint of a process using a resource
with that (unsafe) type. In this case, the process must wait to be
notified of the (probabilistic) choice that has occurred before
using the resource in a message.  Should the need arise to send an
unsafe endpoint in a message, it is possible to patch the endpoint's
session type so as to make it safe, for example by prefixing the
session type with a dummy input/output action. We will see an
instance where this patch is necessary in \cref{ex:threesome}.

\begin{table}
  \framebox[\textwidth]{
    \begin{math}
      \displaystyle
      \begin{array}{@{}c@{}}
        \inferrule[\rulename{t-idle}]{
          \un\Context
        }{
          \wtp\Context\pidle
        }
        \qquad
        \inferrule[\rulename{t-done}]{
          \un\Context
        }{
          \wtp{\Context, \x : \tdone}\pdone\x
        }
        \qquad
        \inferrule[\rulename{t-var}]{
          \un\Context
          \\
          \A : \seqof\t
          \\
          \safe{\seqof\t}
        }{
          \wtp{\Context, \seqof{\x : \t}}{\pvar\A{\seqof\x}}
        }
        \\\\
        \inferrule[\rulename{t-in}]{
          \wtp{\Context, \x : \T, \y : \t}\P
        }{
          \wtp{\Context, \x : \cin\t\T}{\pin\x\y\P}
        }
        \qquad
        \inferrule[\rulename{t-branch}]{
          \wtp{\ContextA, \x : \T}\P
          \\
          \wtp{\ContextB, \x : \S}\Q
        }{
          \wtp{\csum[\p]\ContextA\ContextB, \x : \tbranch\T\S}{\pbranch\x\P\Q}
        }
        \\\\
        \inferrule[\rulename{t-out}]{
          \wtp{\Context, \x : \T}\P
          \\
          \safe\t
        }{
          \wtp{\Context, \x : \cout\t\T, \y : \t}{\pout\x\y\P}
        }
        \qquad
        \inferrule[\rulename{t-left}]{
          \wtp{\Context, \x : \T}\P
        }{
          \wtp{\Context, \x : \tchoice[1]\T\S}{\pleft\x\P}
        }
        \qquad
        \inferrule[\rulename{t-right}]{
          \wtp{\Context, \x : \S}\P
        }{
          \wtp{\Context, \x : \tchoice[0]\T\S}{\pright\x\P}
        }
        \\\\
        \inferrule[\rulename{t-par}]{
          \wtp{\ContextA, x : T}\P
          \\
          \wtp{\ContextB, x : \dualof\T}\Q
        }{
          \wtp{\ContextA, \ContextB, x : \tsession{\pr\T}}{\P \ppar \Q}
        }
        \qquad
        \inferrule[\rulename{t-choice}]{
          \wtp\ContextA\P
          \\
          \wtp\ContextB\Q
        }{
          \wtp{\csum\ContextA\ContextB}{\csum\P\Q}
        }
        \qquad
        \inferrule[\rulename{t-new}]{
          \wtp{\Context, \x : \tsession\p}\P
        }{
          \wtp\Context{\pnew\x\P}
        }
      \end{array}
    \end{math}
  }
  \caption{\label{tab:rules} Typing rules.}
\end{table}

Judgments have the form $\wtp\Context\P$, meaning that $\P$ is well
typed in $\Context$, and are derived by the rules in
\cref{tab:rules}. \REVISION{We assume a global map from process
  variables to sequences of types written
  $\set{A_i : \seqof\t_i}_{i\in I}$ whose domain includes all the
  process variables for which there is a definition
  $\pdef{\A_i}{\seqof\x}{\P_i}$ and that $\wtp{\seqof{x : t}}{\P_i}$
  is derivable for every $i\in I$. This ensures that all process
  definitions are typed consistently. }
The typing rules are syntax directed, so that each process form
corresponds to a typing rule. We now discuss each rule in detail.
Rules \rulename{t-idle} and \rulename{t-done} deal with terminated
processes. In \rulename{t-idle} the whole context must be
unrestricted, since the $\pidle$ process does not use any
resource. Rule \rulename{t-done} is similar, except that the session
$\x$ flagged by the process must have type $\tdone$. This way, we
enforce the correspondence between successful termination in
processes and successful termination in session types.
Rule \rulename{t-var} establishes that a process invocation is well
typed provided that the type of the parameters passed to the process
match the expected ones and that any unused resource has an
unrestricted type. \REVISION{The premise $\A : \seqof\t$ indicates
  that $\A$ is associated with the sequence of types $\seqof\t$ in
  the global map, ensuring that $\A$ is invoked with parameters of
  the right type.}
Observe that the type of such parameters must be safe. This way, we
prevent to use as parameters resources whose type can be (passively)
affected by a probabilistic choice.
Rules~\rulename{t-in} and \rulename{t-out} deal with the exchange of
a message $\y$ on session $\x$. The rules update the type of $\x$
from the conclusion to the premise of the rule to account for the
communication. As usual, a linear resource $\y$ being sent in a
message is no longer available in the continuation of the
process. As anticipated earlier, \rulename{t-out} requires the type
of $\y$ to be safe, again to ensure that the type of $\y$ does not
suddenly change under the effect of a probabilistic choice.

The typing rules described so far are fairly standard for any
session calculus. We now move on to the part of the type system that
handles probabilities.
Rules \rulename{t-left} and \rulename{t-right} deal with
selections. In these cases, the type of $\x$ must be of the form
$\tchoice[p]\T\S$ and the process continuation uses $x$ according to
either $T$ or $S$ respectively. The key aspect is the probability
$p$ with which the process selects ``left'', which is $1$ in the
case of $\pleft[]\x$ and $0$ in the case of $\pright[]\x$. These
processes behave deterministically, hence the probability annotation
in the session type is trivial.
Rule \rulename{t-branch} illustrates the typing of a branch, whereby
a process receives a choice from a session $x$ and continues
accordingly. The type of $x$ must be of the form $\tbranch[p]\T\S$,
where $p$ is the probability with which the process will receive a
``left'' choice. The key part of the rule concerns \emph{all the
  other resources} used by the process, which will be used according
to $\ContextA$ if the process receives a ``left'' choice and
according to $\ContextB$ otherwise. That is, the process is becoming
aware of a probabilistic choice that has been performed elsewhere
and whose outcome is communicated on $x$. Depending on this
information, the process uses its resources (not just $x$)
accordingly. The behavior of the process as a whole is described by
the combination $\csum[p]\ContextA\ContextB$ of the contexts in the
two branches. Recall that the $\csum{}{}$ operator, when used on
session types, is idempotent in all cases but for selections
(\cref{def:ccomb}). Hence, $\csum[p]\ContextA\ContextB$ is
\emph{nearly the same} as $\ContextA$ and $\ContextB$, except that
the probabilities with which some future selections will be
performed on endpoints in $\ContextA$ and $\ContextB$ may have been
adjusted as a side effect of the information received from
$\x$. This mechanism enables the propagation of probabilistic
choices through the system as messages are exchanged on sessions.

Rule \rulename{t-par} deals with parallel compositions
$\P \ppar \Q$, where $\P$ and $Q$ must use $x$ according to dual
session types. Writing $\ContextA,\ContextB$ in the conclusion of
the rule ensures that $P$ and $Q$ do not share any name other than
$x$, thus preventing the creation of network topologies that may
lead to deadlocks~\cite{CairesPfenningToninho16}. In the conclusion
of the rule the type of $x$ becomes of the form $\tsession\p$ to
record the fact that both endpoints of $x$ have been used. The
success probability $p$ coincides with that of one of the endpoints
and is well defined since $\pr\T = \pr{\dualof\T}$.
Rule \rulename{t-choice} deals with probabilistic choices performed
by a process and partially overlaps with \rulename{t-branch} in that
the contexts of the two alternative evolutions of the process after
the choice are combined by $\csum{}{}$.
Finally, rule \rulename{t-new} removes a session $x$ from the
context when $x$ is restricted.

Let us now discuss the main properties enjoyed by well-typed
processes. First and foremost, typing is preserved by reductions.

\begin{restatable}[subject reduction]{theorem}{restatesr}
  \label{thm:sr}
  If $\wtp\Context\P$ and $\P \lred \Q$, then $\wtp\Context\Q$.
\end{restatable}

Although this result is considered standard, one detail makes it
special in our setting. Specifically, we observe that the reduct $Q$
is well typed in the \emph{very same environment} used for typing
$P$, despite the fact that a communication may have taken place on a
session $x$ in $P$, determining a change in the session types
associated with the endpoints of $x$. A communication can occur only
if $P$ contains \emph{both} endpoints for $x$, and more precisely if
there are two subprocesses of $P$ that use $x$ according to dual
session types and that are composed in parallel using
\rulename{t-par}. Then, $x$ in $\Context$ must be associated with a
type of the form $\tsession\p$, where $p$ is the success probability
of $P$. Then, \cref{thm:sr} guarantees that not only the typing, but
also the \emph{success probability of sessions is preserved by
  reductions}. This is counterintuitive at first, given that a
session may evolve through different branches each having different
success probabilities. However, recall that probabilistic choices
are \emph{persistent} in our calculus, meaning that the reduct $Q$
accounts for \emph{all possible evolutions} of $P$. This is what
entails such strong formulation of \cref{thm:sr}.

Next we turn our attention to termination. To this aim, we provide
two characterizations of process termination respectively concerning
the present and the future states of a process.

\begin{definition}[immediate and eventual termination]
  \label{def:terminates}
  We say that $\P$ is \emph{terminated} if $\terminated\P$ is
  derivable using the following axioms and rules:
  \[
    \terminated\pidle
    \qquad
    \terminated{\pdone\x}
    \qquad
    \inferrule{
      \terminated\P
      \\
      \terminated\Q
    }{
      \terminated{\P \ppar \Q}
    }
    \qquad
    \inferrule{
      \terminated\P
      \\
      \terminated\Q
    }{
      \terminated{\csum[\p]\P\Q}
    }
    \qquad
    \inferrule{
      \terminated\P
    }{
      \terminated{\pnew\x\P}
    }
  \]

  We say that $\P$ \emph{terminates with probability $p$}, notation
  $\terminates\P\p$, if there exist $(P_n)$, $(Q_n)$ and $(p_n)$ for
  $n\in\mathbb{N}$ such that $P \wred \csum[p_n]{P_n}{Q_n}$ and
  $\terminated{\P_n}$ for every $n\in\mathbb{N}$ and
  $\lim_{n\to\infty} p_n = p$.
\end{definition}

In words, $\terminated\P$ means that $P$ does not contain any
pending communications, whereas $\terminates\P\p$ means that $P$
evolves with probability $p$ to states in which there are no pending
communications.
Our type system is not strong enough to guarantee (probable)
termination. \REVISION{For example, the process $\Omega$ defined by
  $\pdef\Omega{}\Omega$ is well-typed and diverges. In general,
  however, well-typed processes are guaranteed to be deadlock free,
  as stated formally below}.

\begin{restatable}[deadlock freedom]{theorem}{restatedf}
  \label{thm:df}
  If $\wtp\EmptyContext\P$ and $\P \wred \Q$, then either $Q \red$
  or $\terminated\Q$.
\end{restatable}

Note that deadlock freedom is not simply a bonus feature of our type
system. It is actually a requirement for proving the properties of
the type system that specifically pertain probabilities, which we
will discuss shortly.  Before doing so, we need an operational
characterization of successful termination relative to a particular
session.

\begin{definition}[successful termination of a session]
  \label{def:success}
  We say that $\P$ \emph{successfully terminates session $\x$ with
    probability $\p$} if $\P \success\x\p$ is derivable using the
  following axioms and rules:
  \[
    \inferrule[\rulename{p-done}]{
      \mathstrut
    }{
      \pdone\x \success\x{1}
    }
    \qquad
    \inferrule[\rulename{p-par-1}]{
      \P \success\x\p
    }{
      \P \ppar \Q \success\x\p
    }
    \qquad
    \inferrule[\rulename{p-par-2}]{
      \Q \success\x\p
    }{
      \P \ppar \Q \success\x\p
    }
    \qquad
    \inferrule[\rulename{p-res}]{
      \P \success\x\p
      \\
      \x \ne \y
    }{
      \pnew[q]\y\P \success\x\p
    }
    \qquad
    \inferrule[\rulename{p-choice}]{
      \P \success\x\q
      \\
      \Q \success\x\r
    }{
      \csum[\p]\P\Q \success\x{\p\q + (1-p)r}
    }
    \qquad
    \inferrule[\rulename{p-any}]{
      \mathstrut
    }{
      \P \success\x{0}
    }
  \]
\end{definition}

Axiom \rulename{p-done} states that a process of the form $\pdone\x$
has successfully terminated session $\x$ with probability 1.
The rules~\rulename{p-par-$i$} state that the successful termination
of a parallel composition $P \ppar Q$ with respect to a session $\x$
can be reduced to the successful termination of either $P$ or
$Q$. In particular, we do not require that \emph{both} $P$ and $Q$
have successfully terminated $x$, for two reasons: first, it could
be the case that $P$ and $Q$ are connected by a session different
from $x$, hence only one among $P$ and $Q$ could own $x$; second, if
a process has successfully terminated a session through one of its
endpoints, then duality ensures that the peer owning the other
endpoint cannot have pending operations on it, so the session as a
whole can be considered successfully terminated even if only one
peer has become $\pdone\x$.

Rule~\rulename{p-res} accounts for session restrictions in the
expected way and \rulename{p-choice} states that the successful
termination of $\x$ in a process distribution is obtained by
weighing the probabilities of successful termination of the
processes in the distribution. Note that \rulename{p-choice} can be
applied only if it is possible to derive the successful termination
of $x$ for \emph{all} of the processes in the distribution, whereas
in general only \emph{some} of such processes will have successfully
terminated $x$. To account for this possibility, we can use
\rulename{p-any} to \emph{approximate} the probability of successful
termination of $x$ for any process to 0.

The type system gives us an upper bound to the success probability
of any session:

\begin{proposition}
  \label{prop:sound}
  If $\wtp{x : \tsession\p}\P$ and $\P \success\x\q$, then
  $\q \leq \p$.
\end{proposition}

In particular, a session with type $\tsession0$ cannot be
successfully completed, which could indicate a flaw in the system.
The upper bound is matched exactly by terminated processes:

\begin{restatable}{theorem}{restatesound}
  \label{thm:soundness}
  If $\wtp{x : \tsession\p}\P$ and $\P \nred$, then
  $\P \success\x\p$.
\end{restatable}

Note that \cref{thm:soundness} does not hold unless processes are
deadlock free, whence the key role of \cref{thm:df}.
As stated, \cref{thm:soundness} appears of limited use since it only
concerns processes that cannot reduce any further, whereas in
general we are interested in computing the probability of successful
termination also for processes engaged in arbitrarily long
interactions, for which the predicate $\P \nred$ might never
hold. It turns out that \cref{thm:soundness} can be relativized to
the probability that a process terminates, thus:

\begin{restatable}[relative success]{corollary}{restaterelative}
  \label{cor:relative.success}
  Let $P \lsuccess\x\p$ if there exist $(P_n)$ and $(p_n)$ such that
  $P \wred P_n$ and $P_n \success\x{p_n}$ for all $n\in\mathbb{N}$
  and $\lim_{n\to\infty} p_n = p$.  Then (1)
  $\wtp{x : \tsession1}\P$ and $\terminates\P\p$ imply
  $P \lsuccess\x\p$ and (2) $\wtp{x : \tsession\p}\P$ and
  $\terminates\P1$ imply $P \lsuccess\x\p$.
\end{restatable}

Property (1) states that a well-typed process using a session with
type $x : \tsession{1}$ successfully completes the session with the
same probability with which it terminates. Property (2) extends
\cref{thm:soundness} to processes that are known to terminate with
probability 1.

\begin{example}
  \label{ex:auction.typing}
  Below is the type derivation for the process $\Buyer$ from
  \cref{ex:auction} using $T$ from \cref{ex:auction.buyer} and
  assuming the type assignment $\Buyer : \T$.

  \medskip
  \hspace{-6mm}
  \begin{prooftree}
    \[
      \[
        \justifies
        \wtp{
          x : \tdone
        }{
          \pdone\x
        }
        \using\rulename{t-done}
      \]
      \hspace{-5mm}
      \[
        \[
          \[
            \[
              \justifies
              \wtp{
                x : \tend
              }{
                \pidle
              }
              \using\rulename{t-idle}
            \]
            \justifies
            \wtp{
              x : \tchoice[1]\tend\T
            }{
              \pleft[]\x
            }
            \using\rulename{t-left}
          \]
          \[
            \[
              \justifies
              \wtp{
                x : \T,
                y : \tint
              }{
                \pvar\Buyer\x
              }
              \using\rulename{t-var}
            \]
            \justifies
            \wtp{
              x : \tchoice[0]\tend\T,
              y : \tint
            }{
              \pright\x\pvar\Buyer\x
            }
            \using\rulename{t-right}
          \]
          \justifies
          \wtp{
            x : \tchoice[q]\tend\T,
            y : \tint
          }{
            \csum[q]{
              \pleft[]\x
            }{
              \pright\x\pvar\Buyer\x
            }
            \using\rulename{t-choice}
          }
        \]
        \justifies
        \wtp{
          x : \cin\tint(\tchoice[q]\tend\T)
        }{
          \pin\x\y
          \parens{
            \csum[q]{
              \pleft[]\x
            }{
              \pright\x\pvar\Buyer\x
            }
          }
        }
        \using\rulename{t-in}
      \]
      \justifies
      \wtp{
        x : \tbranch[p]\tdone{\cin\tint(\tchoice[q]\tend\T)}
      }{
        \pbranch\x{
          \pdone\x
        }{
          \pin\x\y
          \parens{
            \csum[q]{
              \pleft[]\x
            }{
              \pright\x\pvar\Buyer\x
            }
          }
        }
      }
      \using\rulename{t-branch}
    \]
    \justifies
    \wtp{
      x : \T
    }{
      \pout\x{\textit{bid}}
      \pbranch\x{
        \pdone\x
      }{
        \pin\x\y
        \parens{
          \csum[p]{
            \pleft[]\x
          }{
            \pright\x\pvar\Buyer\x
          }
        }
      }
    }
    \using\rulename{t-out}
  \end{prooftree}
  \medskip

  Observe the application of \rulename{t-choice}, which turns the
  probabilistic choice $\tchoice[q]\tend\T$ in the conclusion of the
  rule into a deterministic one in the two premises.
  There exists an analogous derivation for $\wtp{x : \dualof\T}\Q$
  where $Q$ is the body of $\pvar\Seller\x$ in \cref{ex:auction}. By
  taking $p$ and $q$ as in \cref{ex:auction.buyer.full}, we derive
  $\wtp{x : \tsession{\frac13}}{\pvar\Buyer\x \ppar \pvar\Seller\x}$
  with one application of \rulename{t-par}. It is easy to establish
  that this process terminates with probability 1, hence by
  \cref{cor:relative.success}(2) the buyer wins the auction with
  probability $\frac13$.
  \eoe
\end{example}

\begin{example}
  \label{ex:typing-choices}
  The separation of probabilistic choices from the communication of
  information (``left'' and ``right'' selections) that depends on
  such choices implies that there is no 1-to-1 correspondence
  between choices as seen in session types and choices performed by
  processes. Below are a few instances in which the type system
  performs a non-trivial reconciliation between the probability
  annotations in types and those in processes. The type derivations
  are detailed in \cref{sec:typing-examples}.
  \begin{enumerate}
  \item The process
    $\pbranch\x{\pright\y\pdone\x}{\pleft\y\pdone\y}$ inverts a
    choice from session $x$ to $y$, so that it successfully
    completes $x$ if and only if it does not successfully complete
    $y$. It is well typed in the context
    $\x : \tbranch[p]\tdone\tend, \y : \tchoice[1-p]\tdone\tend$,
    which reflects the effect of the inversion.
  \item The process
    $\pbranch\x{ \pbranch\y{ \pleft\z\pdone\z }{ \pright[]\z } }{
      \pbranch\y{ \pright[]\z }{ \pright[]\z } }$ coalesces two
    choices received from $x$ and $y$ into a choice sent on $z$. The
    process is well typed in the context
    $\x : \tbranch[p]\tend\tend, \y : \tbranch[q]\tend\tend, \z :
    \tchoice[pq]\tdone\tend$, indicating that the success
    probability for $z$ is the product of the probabilities of
    receiving ``left'' from both $x$ and $y$.\Luca{Qui non abbiamo
      insistito sul fatto che il prodotto $pq$ \`e giusto solo
      sapendo che le due scelte sono indipendenti. Questo \`e
      garantito dal type system perch\'e se il processo parla sia
      con $x$ e $y$, allora $x$ e $y$ sono indipendenti. Hern\'an ha
      giustamente accennato a questo fatto nei related, ma
      idealmente andrebbe rimarcato anche altrove perch\'e mostra
      una interazione importante tra tipaggio e probabilit\`a.}
  \item The process
    $\csum[\frac12]{\pleft\x\pleft\x\pdone\x}{\pright\x\pright[]\x}$
    sends the same probabilistic choice twice on session $x$. It is
    well typed in the context
    $x :
    \tchoice[\frac12]{(\tchoice[1]\tdone\tend)}{(\tchoice[0]\tend\tend)}$
    but \emph{not} in the context
    $x :
    \tchoice[\frac12]{(\tchoice[\frac12]\tdone\tend)}{(\tchoice[\frac12]\tend\tend)}$. Once
    the choice is communicated, subsequent ``left'' or ``right''
    selections that depend on that choice become deterministic.
    \eoe
  \end{enumerate}
\end{example}

\begin{example}[Work sharing]
  \label{ex:threesome}
  Consider a system
  $\pvar{C}\x \ppar \pin\x\z\pvar\B{x,y,z} \ppar \pvar\A\y$ modeling
  (from left to right) a master process $C$ connected with two slave
  processes which can be ``busy'' handling jobs or ``idle'' waiting
  for jobs. The processes are defined as follows:
  \[
    \begin{array}{r@{~}l}
      \pdef{C}\x{& \pout\x\job\pbranch\x{\pdone\x}\pidle}
      \\
      \pdef\B{x,y,\job}{&
        \pout\y{\angles{}}
        \left(
          \csum[p]{
            \pleft\x\pleft\y\pdone\x
          }{
            \left(
              \csum[q]{
                \pright\x\pleft[]\y
              }{
                \pright\y\pout\y\x\pout\y\job\pvar\A\y
              }
            \right)
          }
        \right)
      }
      \\
      \pdef\A\y{&
        \pin\y{}
        \pbranch\y\pidle{
          \pin\y\x
          \pin\y\z
          \pvar\B{x,y,z}
        }
      }
    \end{array}
  \]

  The master sends a job to the first slave and waits for a
  notification indicating whether the job has been handled or
  not. Obviously, the master succeeds only in the first case.
  A busy slave decides whether to handle the job (with probability
  $p$) or not (with probability $1-p$). In the first case, it
  notifies the master and the idle slave that the job has been
  handled and terminates. In the second case, it decides whether to
  discard the job (with probability $q$) or to hand it over to the
  other slave (with probability $1-q$).
  Note that the busy slave sends on $y$ a dummy value to the idle
  one before taking any decision so that the type of $y$ is
  \emph{safe} when $y$ is used in $\pvar\A\y$.
  This way, by the time the busy slave makes a probabilistic choice
  that may affect (and will be communicated to) the idle slave, the
  idle slave is blocked on a $\mkkeyword{case}$ waiting for such
  choice, and therefore its typing can be suitably adjusted when it
  is moved (by \rulename{s-par-choice}) into the scope of the choice.

  Now, take
  $T =
  \cout\tunit(\tchoice[p-pq+q]\tend{\cout\S\cout\tint}\dualof\T)$
  and $S = \tchoice[r]\tdone\tend$ where $\max\set{p,q} > 0$ and
  $r = \frac\p{p-pq+q}$. It is possible to show that the above
  composition is well typed under the global type assignments
  $C : \cout\tjob\dualof\S$, $\B : \S, \T, \tjob$ and
  $\A : \dualof\T$, where we assume that $\job$ has type $\tint$.
  From the fact that the system terminates with probability $1$, we
  conclude that the master succeeds with probability $r$. Details
  can be found in \cref{sec:typing-threesome}.
  \eoe
\end{example}

%%% Local Variables:
%%% mode: latex
%%% TeX-master: "main"
%%% End:

\section{Related Work}
\label{sec:related}

\subparagraph*{Type systems for probabilistic, concurrent programs.}
Despite their close relationship with process algebras, many of
which have been extensively studied in a probabilistic setting,
there are few results concerning probabilistic variants of session
types.
A notable exception is \cite{DBLP:journals/corr/abs-1909-01748}, which
considers a probabilistic variant of multiparty session types (MST)
  where global types are decorated by ranges of probabilities
  representing the degree of likelihood for interactions to happen.
  Besides using MST while we use binary session types, a key
  difference is that~\cite{DBLP:journals/corr/abs-1909-01748} does not
  consider interleaved sessions.
  The effect of probabilistic choices across different sessions and
  the type system presented therein ensures that the aggregate
  probability of all execution paths is 1, which in our case is
  guaranteed by the semantics of the probabilistic choice operator in
  processes.

The type system in \cite{DBLP:journals/corr/abs-1909-01748}
  essentially checks that each choice in a process is made according
  to the probability range written in its type, i.e., a process
  chooses a branch with a probability value that lies within the range
  specified by its session type.
  Differently, a probability value in our types does not necessarily
  translate into the same probability value in a process; moreover,
  the same probabilistic choice in a process may be reflected as
  different probabilities in different sessions, as illustrated in
  \cref{ex:typing-choices}.

Some type systems for probabilistic programs have been developed to
characterize precisely the space of the possible execution
traces~\cite{DBLP:journals/pacmpl/LewCSCM20} or to ensure that
well-typed programs do not leak secret
information~\cite{DBLP:journals/pacmpl/DaraisSLH20}.
The work~\cite{varacca2007probabilistic} considers a sub-structural
type system for a probabilistic variant of the linear
$\pi$-calculus. Although the type system is not concerned with
probabilities directly, there are interesting analogies with our
typing discipline: it is only by relying on the properties of
well-typed processes -- most notably, race and deadlock freedom --
that we are able to relate the probabilities in processes with those
in types.

\subparagraph*{Probabilistic models of concurrent
  processes.}

The design of computational models that combine concurrency and
probabilities has a long
tradition~\cite{vardi1985automatic,segala1995probabilistic} and gave
birth to a variety of operational
approaches~\cite{sokolova2004probabilistic} and concrete probabilistic
extensions of well-known concurrency models, such as
CCS~\cite{DBLP:conf/ecrts/Hansson92},
CSP~\cite{DBLP:journals/tcs/Lowe95,DBLP:conf/mmb/GeorgievskaA12},
Petri nets~\cite{DBLP:conf/fossacs/BonnetKL14},
Klaim~\cite{DBLP:conf/sac/NicolaLM05}, and name-passing process
calculi~\cite{DBLP:conf/fossacs/HerescuP00,varacca2007probabilistic,norman2007model,goubault2007probabilistic}.
Our language for processes can be seen as the session-based
counterpart of (a synchronous version of) the {\em simple
  probabilistic $\pi$-calculus}~\cite{norman2007model}, which features
both probabilistic and non-deterministic choices.
While non-deterministic choices in~\cite{norman2007model} correspond 
to the standard choice operator ($+$) of the $\pi$-calculus, we adopted a session discipline, and hence a  
choice is realised by communicating a label over a session.

The development of a denotational semantics for languages that
combine non-determinism, concurrency and probabilities has revealed
challenging.
On the one hand, probabilistic choices do not distribute over
non-deterministic ones, \ie, it matters whether the environment
chooses before or after a probabilistic choice is made, as
highlighted in~\cite{varacca2006distributing}.
This observation appears to be reflected in our type system by the
typing rules that require a term to be of a safe type, \eg, when a
session is delegated.
Establishing a precise connection between these two notions may pave
the way for generalisations of our probabilistic type combinator.
On the other hand, probabilistic choices in a system need to be
(probabilistically) independent.
This problem is connected with the well-known {\em confusion
  phenomenon}, in which concurrent (and hence, independent) choices
may influence each other (\eg, one choice may enable/disable some
branch in another choice).
As shown
in~\cite{abbes2006true,katoen2013taming,BruniMM19},
confusion can be avoided by establishing an order in which choices are
executed; essentially, by reducing concurrency.
We remark that the session discipline imposed by our language -- and
rule \rulename{t-par} in particular -- makes all probabilistic
choices independent (in a probabilistic sense).

\subparagraph*{Probabilistic languages and analyses.}
Probabilistic models are frequently used to prove properties that can
be expressed as reachability probabilities; they are then verified by
model-checking~\cite{katoen2016probabilistic}.
Our types are also reachability properties related to the probability
of successful completion of a session. Besides, our type system
guarantees deadlock-freedom.
Many approaches have been recently proposed for reasoning on
probabilistic programs, \eg, deductive-style approaches based on
separation
logic~\cite{DBLP:journals/pacmpl/BatzKKMN19,DBLP:journals/pacmpl/TassarottiH19,DBLP:journals/pacmpl/BartheHL20},
probabilistic strategy logic~\cite{DBLP:conf/ijcai/AminofKMMR19},
proof of
termination~\cite{DBLP:conf/popl/FioritiH15,DBLP:conf/tacas/LengalLMR17},
static analysis~\cite{DBLP:conf/pldi/WangHR18}, and probabilistic
symbolic execution~\cite{DBLP:conf/sigsoft/BorgesFdP15}.
Typing has been used in the sequential setting to ensure almost-sure
termination in a probabilistic lambda
calculus~\cite{lago2019probabilistic}.
Our type system does not ensure termination, but it could form the
basis for a probabilistic termination analysis.

\subparagraph*{Deadlock-free sessions.}
The technique we use for preventing deadlocks, which only addresses
tree-like network topologies, is directly inspired to logic-based
session type systems~\cite{CairesPfenningToninho16}. However, our
probabilitic analysis is independent of the exact mechanism that
enforces deadlock freedom and applies to other type systems relying
on richer type structures~\cite{Padovani14,DardhaGay18}.

\section{Concluding Remarks}
\label{sec:conclusion}

In this work we start the study of a type-based static analysis
technique for reasoning on probabilistic reachability problems in
session-based systems. We relate a probabilistic variant of a
session-based calculus (\cref{sec:model}) with a probabilistic
variant of binary session types (\cref{sec:types}) and establish a
correspondence between probability annotations in processes and
those in types (\cref{sec:rules}). By breaking down a complex system
of communicating processes into sessions, we are able to modularly
infer properties concerning the (probable) evolution of the system
from the much simpler specifications described by session types.

There are many developments that stem from this work addressing both
technical and practical problems. Here we discuss those looking more
promising or intriguing.

To make our approach practical, the type system must be supported by
suitable type checking and inference algorithms. Indeed, even though
the typing rules are syntax directed, the probabilistic type
combinator (\cref{def:ccomb}) is difficult to deal with because it
is not injective (the same type can result from combining types with
different probability annotations). We are also considering
extensions of the very same operator so that it is applicable to
``deep choices'' that do not necessarily occur at the top level of a
session type. This extension requires a careful balancing with the
notion of type safety (\cref{def:safe}).

Subtyping relations for session types~\cite{GayHole05} are important
for addressing realistic programming scenarios. Given the already
established connections between session subtyping and (fair) testing
relations~\cite{LanevePadovani08,CastagnaEtAl09,Padovani13,BernardiHennessy13,Padovani16}
and the extensive literature on probabilistic testing
relations~\cite{CleavelandEtAl99,NunezRuperez99,DengEtAl07,deng2009testing}
and behavioral equivalences~\cite{LopezNunez04}, the investigation
of probabilistic variants of session subtyping has solid grounds to
build upon.
A related problem is that process models that feature both
non-deterministic and probabilistic choices are known to be
difficult to model and analyze~\cite{DengEtAl07}. It could be the
case that session-based systems with both non-deterministic and
probabilistic choices are easier to address thanks to their simpler
structure, as already observed in~\cite{varacca2007probabilistic}.

Our analysis based on probabilistic session types can be extended in
several ways.
For example, it would be interesting to quantify the probability of
(partial) execution traces rather than (or in addition to) the
reachability of \quo{successful states}.

  As remarked in \cref{sec:types}, reachability ensures uniqueness
  of solutions of the systems of equations induced by \cref{def:pr},
  but it could be interesting to analyse the spectra of solutions
  obtained when reachability is dropped.
  One could also study variants where probabilities are allowed to
  vary during the execution.
  For instance, one would like to analyse recursive protocols where
  probabilities may decrease (or increase) at each iteration.
  In our setting this may spoil regularity (subtrees may be decorated
  with infinitely many probabilities), allowing one to give non
  finitary specifications.
  A possible way of tackling this problem is to allow imprecise
  probabilities in the types; this may retain regularity at the cost
  of a coarser static analysis.
Probability ranges could also be useful in those cases where
probability annotations in processes are uncertain, possibly because
they have been estimated from execution traces~\cite{EmamMiller18}.
Besides probabilities, there might be other methods
  suitable to model the uncertainty behind the behavior of
  processes.
Further approaches include the \quo{possibilistic} one,
where uncertainty is described using linguistic categories with fuzzy boundaries~\cite{ZADEH1965338}, information gap decision
theory, where the impact of uncertain parameters is estimated by
the deviation of errors~\cite{ben2006info}, and interval analysis,
where uncertain parameters are modelled as intervals and worst-case
analysis is usually performed~\cite{moore2009introduction}.
We think that probability annotations in session types may also
support forms of static analysis aimed at quantifying the
termination probability of session-based programs. Known type
systems that ensure progress, deadlock and livelock freedom are
often quite constraining on the structure of well-typed
programs~\cite{Padovani14,CoppoEtAl16,BalzerToninhoPfenning19}. It
could be the case that switching to a probabilistic setting broadens
substantially the range of addressable programs.

\bibliography{references}

\appendix

\section{Supplement to \cref{sec:types}}
\label{sec:supplement_types}

\begin{example}
  \label{ex:computing-absorbing}
  Consider the type $\T$ in \cref{ex:auction.buyer.full}.
  The transition matrix $P =[p_{ij}]$ of its associated DTMC is
  shown below:
  \[
    \begin{array}{l@{}l}
      P =
      \left[
        \begin{array}{cc|cccc}
          1 & 0 & 0 & 0 & 0 & 0\\
          0 & 1 & 0 & 0 & 0 & 0\\
          \hline
          0 & 0 & 0 & 1 & 0 & 0\\
          \frac14 & 0 & 0 & 0 & \frac34 & 0\\
          0 & 0 & 0 & 0 & 0 & 1\\
          0 & \frac23 & \frac13 & 0 & 0  & 0\\
        \end{array}
      \right]
      &\qquad {\it where}\qquad
      \begin{array}{l@{\ = \ }l}
        \S_0 & \tdone
        \\
        \S_1 & \tend
        \\
        \S_2 & \T
        \\
        \S_3 &
        \tbranch[\frac14]\tdone{
          \cin\tint\parens{
            \tchoice[\frac23]\tend\T}}
        \\
        \S_4 & \cin\tint\parens{\tchoice[\frac23]\tend\T}
        \\
        \S_5 & \tchoice[\frac23]\tend\T
      \end{array}
    \end{array}
  \]

  Note that we have given $P$ in its \emph{canonical
    form}~\cite{KemenySnell76}, in which we have partitioned $P$ in
  four submatrices with the names and meaning described below in
  clockwise order, starting from the top-left corner of $P$:
  \begin{itemize}
  \item $S$ is the 2-by-2 identity matrix giving the probability
    transitions among the absorbing states. By definition of
    absorbing state, this is an identity matrix.
  \item $O$ is the 2-by-4 matrix giving the probability transitions
    from the absorbing states to the transient states. By
    definition, these probabilities are all zeros.
  \item $Q$ is the 4-by-4 matrix giving the probability transitions
    among the transient states.
  \item $R$ is the 4-by-2 matrix giving the probability transitions
    from the transient states to the absorbing states.
  \end{itemize}

  Now, the probability of $\S_2$ being absorbed by $\S_1$, \ie,
  $\pr\T$, can be obtained from the matrix $B = [b_{ij}]$ which is
  computed as follows:
  \[
    \begin{array}{@{}l@{\ = \ }l@{}}
      B
      &
      (I - Q)^{-1} R
      \\
      &
      \left[
        \begin{array}{@{}rrrr@{}}
          1 & -1 & 0 & 0\\
          0 & 1 & -\frac34 & 0\\
          0 & 0 & 1 & -1\\
          -\frac13 & 0 & 0  & 1\\
        \end{array}
      \right]^{-1}
      \left[
        \begin{array}{@{}rr@{}}
          0 & 0 \\
          \frac14 & 0\\
          0 & 0\\
          0 & \frac23\\
        \end{array}
      \right]
      =
      \left[
        \begin{array}{@{}rrrr@{}}
          \frac43 & \frac43 & 1 & 1\\
          \frac13 & \frac43 & 1 & 1\\
          \frac49 & \frac49 & \frac43 & \frac43\\
          \frac49 & \frac49 & \frac13 & \frac43\\
        \end{array}
      \right]
      \left[
        \begin{array}{@{}rr@{}}
          0 & 0 \\
          \frac14 & 0\\
          0 & 0\\
          0 & \frac23\\
        \end{array}
      \right]
      =
      \left[
        \begin{array}{@{}rr@{}}
          \frac13 & \frac23 \\
          \frac13 & \frac23 \\
          \frac19 & \frac89 \\
          \frac19 & \frac89 \\
        \end{array}
      \right]
    \end{array}
  \]

  Then, the probability of absorption for $\S_2 =\T$ is
  $b_{00}$. Hence, $\pr\T = \frac13$.
  \eoe
\end{example}

\begin{theorem}[\cite{KemenySnell76}]
  Let $P$ be the transition matrix of an absorbing DTMC and $B^*$ be
  the matrix of the absorption probabilities. Then, $PB^* = B^*$.
\end{theorem}

Note that the column $l$ of $B^*$, \ie, $[b_{il}]$ contains the
probabilities of $s_i$ being absorbed by $s_l$.
Consequently, $b_{ll} = 1$ and $b_{il} = 0$ for all absorbing states
$s_i \neq s_l$.
Also, the probability $b_{il}$ for non-absorbing states $s_i$ can be
obtained by solving the system of linear equations corresponding to
$l$-column of $B^*$ in the equality $B^* = P B^*$ , \ie,
\[
  \begin{array}{@{}l@{\ = \ }ll@{}}
    b_{ll} & 1
    \\
    b_{ii} & 0 & \text{ for all absorbing states } s_i\neq s_l
    \\
    b_{il} & \sum_h p_{ih} \times b_{hl} & \text{ for all non-absorbing states }  \s_i
  \end{array}
\]

When considering the DTMCs associated with session types there are
exactly two absorbing states, namely $\tdone$ and $\tend$. Moreover,
we are interested in computing the column in $B^*$ associated with
$\tdone$. If we write $\pr{\S_i}$ in place of $b_il$ when
$\S_l = \tdone$, then the set of linear equations is
\[
  \begin{array}{l@{\ = \ }ll}
    \pr\tdone & 1
    \\
    \pr\tend & 0
    \\
    \pr{S_i} & \sum_h p_{ih} \times \pr{\S_h}
    &\text{ for all } \S_i\not\in\{\tend,\tdone\}
  \end{array}
\]

\begin{example}
  The system of equations for the DTMC in
  \cref{ex:computing-absorbing} is
  \[
    \begin{array}{l@{\ = \ }l}
      \pr\tdone & 1
      \\
      \pr\tend & 0
      \\
      \pr{\T}  & \pr{\S_3} 
      \\
      \pr{\S_3} & \frac14\pr{\tdone}+\frac34\pr{\S_4}
      \\
      \pr{\S_4} & \pr{\S_5}
      \\
      \pr{\S_5} & \frac23\pr{\tend}+\frac13\pr{\T}
    \end{array}
  \]
  Note in particular that the system of equations corresponds
  exactly to the one derived from \cref{def:pr} and its solution is
  $\pr\T = \frac13$, $\pr{\S_3} = \frac13$, $\pr{\S_5} = \frac19$,
  $\pr{\S_5} = \frac19$.
  \eoe
\end{example}

We conclude this section with the proof of \cref{prop:ccomb}.

\restateccomb*
\begin{proof}
  The only interesting case is when $T_1 = \tchoice[q]\T\S$ and
  $\T_2 = \tchoice[r]\T\S$. We have
  \[
    \begin{array}{l@{~}ll@{}}
      \multicolumn{2}{@{}l}{\pr{\csum{T_1}{T_2}}}
      \\
      = & \pr{\csum{(\tchoice[q]\T\S)}{(\tchoice[r]\T\S)}}
      & \text{by definition of $T_1$ and $T_2$}
      \\
      = & \pr{\tchoice[pq+(1-p)r]\T\S} & \text{by definition of $\csum{}{}$}
      \\
      = & (pq+(1-p)r)\pr\T + (1-pq-(1-p)r)\pr\S
      & \text{by definition of $\pr\cdot$}
      \\
      = & pq\pr\T + r\pr\T - pr\pr\T + \pr\S -pq\pr\S - r\pr\S + pr\pr\S
      \\\\
      \multicolumn{2}{@{}l}{p\pr{T_1} + (1-p)\pr{T_2}}
      \\
      = & p\pr{\tchoice[q]\T\S} + (1-p)\pr{\tchoice[r]\T\S}
      & \text{by definition of $T_1$ and $T_2$}
      \\
      = & p(q\pr\T + (1-q)\pr\S) + (1-p)(r\pr\T + (1-r)\pr\S)
      & \text{by definition of $\pr\cdot$}
      \\
      = & \multicolumn{2}{l}{pq\pr\T + p\pr\S - pq\pr\S + r\pr\T + \pr\S - r\pr\S - pr\pr\T - p\pr\S + pr\pr\S}
      \\
      = & \multicolumn{2}{l}{pq\pr\T + r\pr\T - pr\pr\T + \pr\S - pq\pr\S - r\pr\S + pr\pr\S}
    \end{array}
  \]
  which confirms the statement.
\end{proof}

\section{Examples}
\label{sec:typing_work_sharing}

\subsection{Typing of \cref{ex:typing-choices}}
\label{sec:typing-examples}

\begin{enumerate}
\item The derivation below shows that
  $\pbranch\x{\pright\y\pdone\x}{\pleft\y\pdone\y}$ is well typed in
  the context
  $\x : \tbranch[p]\tdone\tend, \y : \tchoice[1-p]\tdone\tend$.
\[
  \begin{prooftree}
    \[
      \[
        \justifies
        \textstyle
        \wtp{
          \x : \tdone,
          \y : \tend
        }{
          \pdone\x
        }
      \using\rulename{t-done}
      \]
      \justifies
      \textstyle
      \wtp{
        \x : \tdone,
        \y : \tchoice[0]\tdone\tend
      }{
        \pright\y\pdone\x
      }
      \using\rulename{t-right}
    \]
    \qquad
    \[
      \[
        \justifies
        \textstyle
        \wtp{
          \x : \tend,
          \y : \tdone
        }{
          \pdone\y
        }
      \using\rulename{t-done}
      \]
      \justifies
      \textstyle
      \wtp{
        \x : \tend,
        \y : \tchoice[1]\tdone\tend
      }{
        \pleft\y\pdone\y
      }
      \using\rulename{t-left}
    \]
    \justifies
    \textstyle
    \wtp{
      \x : \tbranch[p]\tdone\tend,
      \y : \tchoice[1-p]\tdone\tend
    }{
      \pbranch\x{
        \pright\y\pdone\x
      }{
        \pleft\y\pdone\y
      }
    }
  \using\rulename{t-branch}
  \end{prooftree}
\]

\item The following derivation shows that 
$\pbranch\x{ \pbranch\y{ \pleft\z\pdone\z }{ \pright[]\z } }{
      \pbranch\y{ \pright[]\z }{ \pright[]\z } }$ is well typed in the context
    $\x : \tbranch[p]\tend\tend, \y : \tbranch[q]\tend\tend, \z :
    \tchoice[pq]\tdone\tend$.

\[
  \begin{prooftree}
    \[
      \[
        \[
          \justifies
          \textstyle
          \wtp{
            \x : \tend,
            \y : \tend,
            \z : \tdone
          }{
            \pdone\z
          }
	  \using\rulename{t-done}
        \]
        \justifies
        \textstyle
        \wtp{
          \x : \tend,
          \y : \tend,
          \z : \tchoice[1]\tdone\tend
        }{
          \pleft\z\pdone\z
        }
        \using\rulename{t-left}
      \]
      \[
        \[
          \justifies
          \textstyle
          \wtp{
            \x : \tend,
            \y : \tend,
            \z : \tend
          }{
            \pidle
          }
	  \using\rulename{t-idle}
        \]
        \justifies
        \textstyle
        \wtp{
          \x : \tend,
          \y : \tend,
          \z : \tchoice[0]\tdone\tend
        }{
          \pright[]\z
        }
        \using\rulename{t-left}
      \]
      \justifies
      \textstyle
      \wtp{
        \x : \tend,
        \y : \tbranch[q]\tend\tend,
        \z : \tchoice[q]\tdone\tend
      }{
        \pbranch\y{
          \pleft\z\pdone\z
        }{
          \pright[]\z
        }
      }
     \using\rulename{t-branch}
    \]
    \quad
    \vdots
    \justifies
    \textstyle
    \wtp{
      \x : \tbranch[p]\tend\tend,
      \y : \tbranch[q]\tend\tend,
      \z : \tchoice[pq]\tdone\tend
    }{
      \pbranch\x{
        \pbranch\y{
          \pleft\z\pdone\z
        }{
          \pright[]\z
        }
      }{
        \pbranch\y{
          \pright[]\z
        }{
          \pright[]\z
        }
      }
    }
  \using\rulename{t-branch}
  \end{prooftree}
\]

\item We illustrate below that $\csum[\frac12]{\pleft\x\pleft\x\pdone\x}{\pright\x\pright[]\x}$
    cannot be typed with the context
    $x :
    \tchoice[\frac12]{(\tchoice[\frac12]\tdone\tend)}{(\tchoice[\frac12]\tend\tend)}$. 
\[
  \begin{prooftree}
    \[
      \[
        \justifies
        \textstyle
        \wtp{
          \x : \tchoice[\frac12]\tdone\tend
        }{
          \pleft\x
          \pdone\x
        }
        \using \skull
      \]
      \justifies
      \textstyle
      \wtp{
        \x : \tchoice[1]{\parens{\tchoice[\frac12]\tdone\tend}}\tend
      }{
        \pleft\x
        \pleft\x
        \pdone\x
      }
     \using\rulename{t-left} 
    \]
    \[
      \[
        \justifies
        \textstyle
        \wtp{
          \x : \tend
        }{
          \pidle
        }
      \]
      \justifies
      \textstyle
      \wtp{
        \x : \tchoice[0]{\parens{\tchoice[\frac12]\tdone\tend}}\tend
      }{
        \pright[]\x
      }
    \using\rulename{t-right} 
    \]
    \justifies
    \textstyle
    \wtp{
      \x : \tchoice[\frac12]{\parens{\tchoice[\frac12]\tdone\tend}}\tend
    }{
      \csum[\frac12]{
        \parens{
          \pleft\x
          \pleft\x
          \pdone\x
        }
      }{
        \pright[]\x
      }
    }
  \using\rulename{t-choice} 
  \end{prooftree}
\]
\end{enumerate}

\subsection{Typing of \cref{ex:threesome}}
\label{sec:typing-threesome}
 
We first show that the defining equation for the process variable $C$
is well typed, \ie, that the judgement $\wtp{ x :
  \cout\tjob(\tbranch[r]\tdone\tend) }{
  \pout\x\job\pbranch\x{\pdone\x}\pidle }$ holds (when assuming $\job$
is of type $\tint$).

\[
\begin{prooftree}
 \[
  \[
  \justifies \wtp{ x : \tdone }{ \pdone\x } \using\rulename{t-done}
  \]
  \[
  \justifies
  \wtp{
    x : \tend
  }{
    \pidle
  }
  \using\rulename{t-idle}
  \]
  \justifies
  \wtp{
    x : \tbranch[r]\tdone\tend
  }{
    \pbranch\x{\pdone\x}\pidle 
  }
  \using\rulename{t-branch}
  \]
  \safe\tjob
  \justifies
  \wtp{
    x : \cout\tjob(\tbranch[r]\tdone\tend)
  }{
    \pout\x\job\pbranch\x{\pdone\x}\pidle 
  }
  \using\rulename{t-out}
\end{prooftree}
\]

We now consider the defining equation for the process variable $\B$. 
For presentation purposes we consider first  the derivations for  three 
different subterms corresponding to the alternative choices in the definition.
In particular, 
\begin{itemize}
\item
  $\wtp{
  \x : \tchoice[1]\tdone\tend,
  \y : \tchoice[1]\tend{\cout\S\cout\tjob\dualof\T},
  \job : \tjob
}{
  \pleft\x\pleft\y\pdone\x
}$
  \cref{eq-typing-work-a};

\item
  $\wtp{
  \x : \tchoice[0]\tdone\tend,
  \y : \tchoice[1]\tend{\cout\S\cout\tjob\dualof\T},
  \job : \tjob
}{
  \pright\x\pleft[]\y
}$
  \eqref{eq-typing-work-b};

\item
  $\wtp{
  \x : \S,
  \y : \tchoice[0]\tend{\cout\S\cout\tjob\dualof\T},
  \job : \tjob
}{
  \pright\y\pout\y\x\pout\y\job\pvar\A\y
}$
  \eqref{eq-typing-work-c}.
\end{itemize}

\begin{equation}
  \label{eq-typing-work-a}
  \begin{prooftree}
    \[
    \[
    \justifies
    \wtp{
      \x : \tdone,
      \y : \tend,
      \job : \tjob
    }{
      \pdone\x
    }
    \using\rulename{t-done}
    \]
    \justifies
    \wtp{
      \x : \tdone,
      \y : \tchoice[1]\tend{\cout\S\cout\tjob\dualof\T},
      \job : \tjob
    }{
      \pleft\y\pdone\x
    }
    \using\rulename{t-left}
    \]
    \justifies
    \textstyle
    \wtp{
      \x : \tchoice[1]\tdone\tend,
      \y : \tchoice[1]\tend{\cout\S\cout\tjob\dualof\T},
      \job : \tjob
    }{
      \pleft\x\pleft\y\pdone\x
    }
    \using\rulename{t-left}
  \end{prooftree}
\end{equation}

\begin{equation}
  \label{eq-typing-work-b}
  \begin{prooftree}
    \[\[
    \justifies
    \textstyle
    \wtp{
      \x : \tend,
      \y : \tend,
      \job : \tjob
    }{
      \pidle
    }
    \using\rulename{t-idle}
    \]
    \justifies
    \textstyle
    \wtp{
      \x : \tend,
      \y : \tchoice[1]\tend{\cout\S\cout\tjob\dualof\T},
      \job : \tjob
    }{
      \pleft[]\y
    }
    \using\rulename{t-left}
    \]
    \justifies
    \textstyle
    \wtp{
      \x : \tchoice[0]\tdone\tend,
      \y : \tchoice[1]\tend{\cout\S\cout\tjob\dualof\T},
      \job : \tjob
    }{
      \pright\x\pleft[]\y
    }
    \using\rulename{t-right}
  \end{prooftree}
\end{equation}

\begin{equation}
  \label{eq-typing-work-c}
  \begin{prooftree}
  \[
  \[
  \[
  \A : \dualof\T \qquad \safe{\dualof\T}
  \justifies 
    \textstyle
    \wtp{
      \y : \dualof\T
    }{
      \pvar\A\y
    }
    \using\rulename{t-var}
    \] 
  \safe\tjob
  \justifies
  \textstyle
  \wtp{
    \y : \cout\tjob\dualof\T,
    \job : \tjob
  }{
    \pout\y\job\pvar\A\y
  }
  \using\rulename{t-out}     
  \]
  \safe\S
  \justifies 
  \textstyle
  \wtp{
    \x : \S,
    \y : \cout\S\cout\tjob\dualof\T,
    \job : \tjob
  }{
      \pout\y\x\pout\y\job\pvar\A\y
  }
  \using\rulename{t-out}
  \]
  \justifies
  \textstyle
  \wtp{
    \x : \S,
    \y : \tchoice[0]\tend{\cout\S\cout\tjob\dualof\T},
    \job : \tjob
  }{
    \pright\y\pout\y\x\pout\y\job\pvar\A\y
  }
  \using\rulename{t-right}
  \end{prooftree}
\end{equation}

Then, the derivation for the right-most probabilistic choice in the definition
of $\B$ is obtained
from \cref{eq-typing-work-b} and \cref{eq-typing-work-c} as follows.

\begin{equation}
  \label{eq-typing-work-d}
  \begin{prooftree}
    \[
    \vdots\quad\eqref{eq-typing-work-b}
    \justifies
    \textstyle
    \wtp{
      \x : \tchoice[0]\tdone\tend,
      \y : \tchoice[1]\tend{\cout\S\cout\tjob\dualof\T},
      \job : \tjob
    }{
      \ldots
    }
  \]
  \[
  \vdots\quad\eqref{eq-typing-work-c}
  \justifies
  \textstyle
  \wtp{
    \x : \S,
    \y : \tchoice[0]\tend{\cout\S\cout\tjob\dualof\T},
    \job : \tjob
  }{
    \ldots
  }
  \]
  \justifies
  \textstyle
  \wtp{
    \x : \tchoice[(1-\q)\r]\tdone\tend,
    \y : \tchoice[\q]\tend{\cout\S\cout\tjob\dualof\T},
    \job : \tjob
  }{
    \csum[q]{
      \pright\x\pleft[]\y
    }{
      \pright\y\pout\y\x\pout\y\job\pvar\A\y
    }
  } 
  \using\rulename{t-choice}
  \end{prooftree}
\end{equation}

The derivation for the definition of $\B$ is obtained as follows.

\begin{equation}
  \label{eq-typing-work-e}
  \begin{prooftree}
    \[
    \[
    \vdots\quad\eqref{eq-typing-work-a}
    \justifies
    \textstyle
    \wtp{
      \x : \tchoice[1]\tdone\tend,
      \y : \tchoice[1]\tend{\cout\S\cout\tjob\dualof\T},
      \job : \tjob
    }{
      \ldots
    }
    \]
    \[
    \vdots\quad\eqref{eq-typing-work-c}
    \justifies
    \textstyle
    \wtp{
      \x : \tchoice[(1-\q)\r]\tdone\tend,
      \y : \tchoice[\q]\tend{\cout\S\cout\tjob\dualof\T},
      \job : \tjob
    }{
      \ldots
    }
    \]
    \justifies
    \textstyle
    \wtp{
      \x : \tchoice[\p + (1-\q)(1-\q)\r]\tdone\tend,
      \y : \tchoice[\p + (1-\p)\q]\tend{\cout\S\cout\tjob\dualof\T},
      \job : \tjob
    }{
      \csum[\p]\ldots\ldots
    }
    \using\rulename{t-choice}
    \]
    %  \safe\tunit
    \justifies
    \textstyle
    \wtp{
      \x : \tchoice[\p + (1-\q)(1-\q)\r]\tdone\tend,
      \y : \cout\tunit(\tchoice[\p + (1-\p)\q]\tend{\cout\S\cout\tjob\dualof\T}),
      \job : \tjob
    }{
      \pout\y{\angles{}}\csum[\p]\ldots\ldots
    }
    \using\rulename{t-out}
  \end{prooftree}
\end{equation}

The proof is completed by noting that
$p\:+\:\left(1-p\right)\left(1-q\right)\frac{p}{p-pq+q} =
\frac{p}{p-pq+q} = r$, and $\p + (1-\p)\q = \p - \p\q + \q$.

We show that the definition of $\A$ is well typed with the derivation
below.
\[
\begin{prooftree}
  \[
  \[
  \justifies
  \textstyle
  \wtp{
    \y : \tend
  }{
    \pidle
  }
  \using\rulename{t-idle}
  \]
  \[
  \[
  \[
  \B : \S,\T,\tjob \quad \safe{\S,\T,\tjob}
  \justifies 
  \wtp{
    \x : \S,
    \y : \T,
    \z : \tjob
  }{
    \pvar\B{x,y,z}
  }
  \using\rulename{t-var}
  \] 
  \justifies
  \textstyle
  \wtp{
    \x : \S,
    \y : \cin\tjob\T
  }{
    \pin\y\z
    \pvar\B{x,y,z}
  }
  \using\rulename{t-in}     
  \]
  \justifies
  \textstyle
  \wtp{
    \y : \cin\S\cin\tjob\T
  }{
    \pin\y\x
    \pin\y\z
    \pvar\B{x,y,z}
  }
  \using\rulename{t-in}     
  \]
  \justifies 
  \textstyle
  \wtp{
    \y : \tbranch[p-pq+q]{\tend}{\cin\S\cin\tjob\T}
  }{
    \pbranch\y\pidle{
      \pin\y\x
      \pin\y\z
      \pvar\B{x,y,z}
    }
  }
  \using\rulename{t-branch}
  \]
  \justifies
  \textstyle
  \wtp{
    \y : \cin\tunit(\tbranch[p-pq+q]{\tend}{\cin\S\cin\tjob\T})
  }{
    \pin\y{}
    \pbranch\y\pidle{
      \pin\y\x
      \pin\y\z
      \pvar\B{x,y,z}
    }
  }
  \using\rulename{t-in}
\end{prooftree}
\]

The typing for the composition $\pvar{C}\x \ppar
\pin\x\z\pvar\B{x,y,z} \ppar \pvar\A\y$ is obtained as follows.

\[
\begin{prooftree}
  \[
  C : \cout\tjob\dualof\S \quad \safe{\cout\tjob\dualof\S}
  \justifies
  \textstyle
  \wtp{
    \x : \cout\tjob\dualof\S
  }{
    \pvar{C}\x     
  }
  \using\rulename{t-var}
  \]
  \[
  \[
  \[
  B : \S,\T,\tjob \quad \safe{\S,\T,\tjob}
  \justifies
  \textstyle
  \wtp{
    \x : \S,
    \y : \T,
    \z : \tjob
  }{
    \pvar\B{x,y,z} 
  }
  \using\rulename{t-var}
  \]
  \justifies
  \textstyle
  \wtp{
    \x : \cin\tjob\S,
    \y : \T
  }{
    \pin\x\z\pvar\B{x,y,z} 
  }
  \using\rulename{t-in}
  \]
  \[
  A : \dualof\T \quad \safe{\dualof\T}
  \justifies
  \textstyle
  \wtp{
    \y : \dualof\T
  }{
    \pvar\A\y
  }
  \using\rulename{t-var}
  \]
  \justifies
  \textstyle
  \wtp{
    \x : \cin\tjob\S,
    \y : \tsession{\pr\T}
  }{
     \pin\x\z\pvar\B{x,y,z} \ppar \pvar\A\y
  }
  \using\rulename{t-par}
  \]
  \justifies
  \textstyle
  \wtp{
    \x : \tsession{\pr{\cin\tjob\S}},
    \y : \tsession{\pr\T}
  }{
    \pvar{C}\x \ppar \pin\x\z\pvar\B{x,y,z} \ppar \pvar\A\y    
  }
  \using\rulename{t-par}
\end{prooftree}
\]

Finally, we compute the success probabilities:
\begin{itemize}
\item
  $\pr{\cin\tjob\S} = \pr\S = \r\pr\tdone + (1 -\r)\pr\tend = \r$,
  and
\item $\pr\T = 0$ since $\T$ cannot reach $\tdone$. The complete
  computation is as follows.
  \[
  \begin{array}{l@{\ =\ }ll}
    \pr\T & \pr{\tchoice[p-pq+q]\tend{\cout\S\cout\tint}\dualof\T} \\
    &  (\p -\p\q + \q)\pr\tend + \r \pr{\cout\S\cout\tint\dualof\T} & \text{where}\ \r = (1 - (\p -\p\q + \q)) \\
    & \r\pr{\cout\S\cout\tint\dualof\T} &  \text{by}\ \pr\tend = 0 \\
    & \r\pr{\cout\tint\dualof\T}\\
    & \r\pr{\dualof\T}\\
    & \r\pr{\tbranch[p-pq+q]\tend{\cin\S\cin\tint}\T}\\
    & \r(\p -\p\q + \q)\pr\tend + \r^2 \pr{\cin\S\cin\tint\T} \\
    & \r^2\pr{\cin\S\cin\tint\T} &  \text{by}\ \pr\tend = 0 \\
    & \r^2\pr{\cin\tint\T}\\
    & \r^2\pr{T}\\
  \end{array}
  \]
  whose unique solution is $\pr\T = 0$ (for $0<\p,\q<1$).
\end{itemize}

\section{Proof of Theorem~\ref{thm:sr}}
\label{sec:proofs_sr}

\begin{lemma}
  \label{lem:no-choice}
  If $\csum[1]\t\s$ is defined, then $\csum[1]\t\s = \t$.
\end{lemma}
\begin{proof}
  The only interesting case is when $t \ne s$ and this can happen in
  two cases only.
  If $t = \tsession\p$ and $s = \tsession\q$, then we conclude
  $\csum[1]\t\s = \tsession\p$.
  If $t = \tchoice[p]\T\S$ and $s = \tchoice[q]\T\S$, then we
  conclude $\csum[1]\t\s = \tchoice[p]\T\S$.
\end{proof}

The next result shows that, if the very same process can be typed in
two different contexts, then the success probabilities of the
session types in the two contexts is the same. In general it is not
true that the session types themselves are the same, because
\rulename{t-left} and \rulename{t-right} allow selections to be
typed differently as far as the non-selected branch is concerned.
Let $\probeq$ be the smallest equivalence relation on types such
that $\T \probeq \S$ if $\pr\T = \pr\S$. We write
$\ContextA \probeq \ContextB$ if $\ContextA(x) \probeq \ContextB(x)$
for every $x\in\dom\ContextA \cap \dom\ContextB$.

\begin{lemma}
  \label{lem:choice-idem}
  If $\wtp{\Context_i}\P$ for $i=1,2$ and
  $\dom{\Context_1} = \dom{\Context_2}$, then
  $\Context_1 \probeq \Context_2$.
\end{lemma}
\begin{proof}
  By induction on the structure of $P$ and by cases on its shape. We
  only discuss a few representative cases, the others being similar
  or simpler.

  \proofcase{$P = \pidle$}
  Then $\un{\Context_i}$ for $i=1,2$ and we conclude
  $\Context_1 \probeq \Context_2$ by observing that types of the
  form $\tsession\p$ are not unrestricted and that the only
  unrestricted session type is $\tend$.

  \proofcase{$P = \pdone\x$}
  Then there exist $\Context_1'$ and $\Context_2'$ such that
  $\Context_i = \Context_i', x : \tdone$ for $i=1,2$. We conclude
  $\Context_1 \probeq \Context_2$ by the same observations made in
  the previous case.

  \proofcase{$P = \pvar\A{\seqof\x}$}
  Then there exist $\Context_1'$ and $\Context_2'$ such that
  $\Context_i = \Context_i', \seqof{x : t}$ and $\un{\Context_i'}$
  for $i=1,2$ and $A : \seqof\t$. We conclude
  $\Context_1 \probeq \Context_2$ by the same observations made in
  the previous cases.

  \proofcase{$P = \pin\x\y\Q$}
  Then there exist $\Context_1'$, $\Context_2'$, $t_1$, $t_2$, $T_1$
  and $T_2$ such that
  $\Context_i = \Context_i', x : \cin{\t_i}{\T_i}$ and
  $\wtp{\Context_i', x : T_i, y : t_i}\Q$ for $i=1,2$.
  Using the induction hypothesis we deduce
  $\Context_1' \probeq \Context_2'$ and $t_1 \probeq t_2$ and
  $T_1 \probeq T_2$. We conclude $\Context_1 \probeq \Context_2$
  since
  $\pr{\cin{t_1}{T_1}} = \pr{T_1} = \pr{T_2} = \pr{\cin{t_2}{T_2}}$.

  \proofcase{$P = \csum{P_1}{P_2}$}
  Then there exist $\Context_{ij}$ for $1\leq i,j\leq 2$ such that
  $\Context_i = \csum{\Context_{i1}}{\Context_{i2}}$ and
  $\wtp{\Context_{ij}}{P_j}$ for $1 \leq i, j \leq 2$.
  Using the induction hypothesis we deduce
  $\Context_{1j} \probeq \Context_{2j}$ for all $j=1,2$.
  We conclude
  $\Context_1 = \csum{\Context_{11}}{\Context_{12}} \probeq
  \csum{\Context_{21}}{\Context_{22}} = \Context_2$.

  \proofcase{$P = \pleft\x\Q$}
  Then there exist $\Context_1'$, $\Context_2'$, $T_1$, $T_2$, $S_1$
  and $S_2$ such that
  $\Context_i = \Context_i', x : \tchoice[1]{T_i}{S_i}$ and
  $\wtp{\Context_i', x : T_i}\Q$ for $i=1,2$.
  Using the induction hypothesis we deduce
  $\Context_1' \probeq \Context_2'$ and $T_1 \probeq T_2$. We
  conclude $\Context_1 \probeq \Context_2$ by observing that
  $\pr{\tchoice[1]{T_1}{S_1}} = \pr{T_1} = \pr{T_2} =
  \pr{\tchoice[1]{T_2}{S_2}}$.

  \proofcase{$P = P_1 \ppar P_2$}
  From \rulename{t-par} we deduce that there exist $\Context_{11}$,
  $\Context_{12}$, $\Context_{21}$, $\Context_{22}$, $T_1$ and $T_2$
  such that
  $\Context_i = \Context_{i1}, \Context_{i2}, x :
  \tsession{\pr{T_i}}$ and $\wtp{\Context_{i1}, x : T_i}{P_1}$ and
  $\wtp{\Context_{i2}, x : \dualof{T_i}}{P_2}$ for $i=1,2$.
  Using the induction hypothesis we deduce
  $\Context_{1j} \probeq \Context_{2j}$ for $j=1,2$ and
  $T_1 \probeq \T_2$ namely $\pr{\T_1} = \pr{T_2}$. We conclude
  $\Context_1 \probeq \Context_2$.
\end{proof}

The next result shows that a process becoming aware of a
probabilistic choice can be typed differently so as to account for
the probabilistic information transmitted with the choice. This is
the key lemma that allows us to deal with
\rulename{s-par-choice}. Note that, as the process may be connected
with other processes through sessions, the information concerning
the probabilistic choice may need to propagate along an arbitrary
number of sessions.

\begin{lemma}
  \label{lem:splitting}
  If $\wtp{\Context, \x : \dualof{\csum[r]{\T_1}{\T_2}}}\P$, then
  there exist $\Context_1$ and $\Context_2$ such that
  $\csum[r]{\Context_1}{\Context_2} = \Context$ and
  $\wtp{\Context_i, \x : \dualof{\T_i}}\P$ for every $i=1,2$.
\end{lemma}
\begin{proof}
  If $T_1 = T_2$ we conclude immediately by taking
  $\Context_1 = \Context_2 = \Context$, so from now on we assume
  $T_1 \ne T_2$ which can happen only when $T_1$ and $T_2$ are a
  choice.
  We proceed by induction on the derivation of
  $\wtp{\Context, \x : \dualof{\csum[r]{\T_1}{\T_2}}}\P$ and by
  cases on the last rule applied. We discuss only interesting cases,
  particularly those compatible with the assumption $T_1 \ne T_2$.

  \proofrule{t-var}
  Then $P = \pvar\A{\seqof\x}$. From \rulename{t-var} we deduce:
  \begin{itemize}
  \item
    $\Context, \seqof{x : t} = \ContextB, x :
    \dualof{\csum[r]{T_1}{T_2}}$;
  \item $\un\ContextB$;
  \item $\A : \seqof\t$;
  \item $\safe{\seqof\t}$.
  \end{itemize}

  Since $T_1$ and $T_2$ are choices, they cannot be
  unrestricted. Therefore, $x$ must be one of the variables in
  $\seqof\x$ and $\dualof{\csum[r]{T_1}{T_2}}$ is one of the types
  in $\seqof\t$. But then $\dualof{\csum[r]{T_1}{T_2}}$ is a branch,
  which is not a safe type according to
  Definition~\ref{def:safe}. We conclude that this case is
  impossible.

  \proofcase{\rulename{t-branch} when $\x$ is the endpoint being
    used for input}
  Then $\P = \pbranch\x{\P_1}{\P_2}$. From \rulename{t-branch} we
  deduce that there exist $\ContextB_1$, $\ContextB_2$, $S_1$ and
  $S_2$ such that:
  \begin{itemize}
  \item $\csum{\ContextB_1}{\ContextB_2} = \Context$;
  \item $\dualof{\csum[r]{\T_1}{\T_2}} = \tbranch{\S_1}{S_2}$;
  \item $\wtp{\ContextB_i, \x : \S_i}{\P_i}$ for $i=1,2$.
  \end{itemize}

  From Definition~\ref{def:ccomb} we deduce that there exist $\p_1$
  and $\p_2$ such that $\dualof{\T_i} = \tbranch[p_i]{\S_1}{\S_2}$
  and $p = rp_1 + (1-r)p_2$.
  Let $\Context_i \eqdef \csum[p_i]{\ContextB_1}{\ContextB_2}$ and
  observe that
  $\csum{(\csum[p_1]{\ContextB_1}{\ContextB_2})}{(\csum[p_2]{\ContextB_1}{\ContextB_2})}
  = \Context$.  We conclude
  $\wtp{\Context_i, \x :
    \tbranch[p_i]{\S_1}{\S_2}}{\pbranch\x{\P_1}{\P_2}}$ with an
  application of \rulename{t-branch}.

  \proofrule{t-par}
  Then $P = Q \ppar R$.
  Since $\dualof{\csum[r]{T_1}{T_2}}$ is a session type and not a
  type of the form $\tsession\q$, $\x$ cannot be used by both $Q$
  and $R$.  We consider only the case in which $\x$ is used by $Q$,
  the other case being symmetric.
  From \rulename{t-par} we deduce:
  \begin{itemize}
  \item $\wtp{\ContextB_1, y : S, x : \dualof{\csum[r]{T_1}{T_2}}}{Q}$;
  \item $\wtp{\ContextB_2, y : \dualof\S}{R}$;
  \item $\Context = \ContextB_1, \ContextB_2, y : \tsession{\pr\S}$.
  \end{itemize}

  Using the induction hypothesis we deduce that there exist
  $\ContextB_{11}$, $\ContextB_{12}$, $S_1$ and $S_2$ such that
  $\csum[r]{(\ContextB_{11}, y : S_1)}{(\ContextB_{12}, y : S_2)} =
  \ContextB_1, y : S$ and
  $\wtp{\ContextB_{1i}, y : S_i, x : \dualof{T_i}}{Q}$ for $i=1,2$.
  In particular, we have $\dualof\S = \dualof{\csum[r]{S_1}{S_2}}$.
  Using the induction hypothesis once again, we deduce that there
  exist $\ContextB_{21}$ and $\ContextB_{22}$ such that
  $\csum[r]{\ContextB_{21}}{\ContextB_{22}} = \ContextB_2$ and
  $\wtp{\ContextB_{2i}, \y : \dualof{S_i}}$ for $i=1,2$.
  Let
  $\Context_i \eqdef \ContextB_{1i}, \ContextB_{2i}, y :
  \tsession{\pr{S_i}}$ and observe that
  $\csum[r]{\Context_1}{\Context_2} = \Context$.  We conclude
  $\wtp{\Context_i, x : \dualof{T_i}}\P$ for $i=1,2$ using
  \rulename{t-par}.

  \proofrule{t-choice}
  Then we have:
  \begin{itemize}
  \item $P = \csum{P_1}{P_2}$ for some $P_1$ and $P_2$;
  \item $\dualof{\csum[r]{T_1}{T_2}} = \csum{S_1}{S_2}$ for some
    $S_1$, $S_2$ and $p$;
  \item $\csum{\ContextB_1}{\ContextB_2} = \Context$ for some
    $\ContextB_1$ and $\ContextB_2$;
  \item $\wtp{\ContextB_k, \x : S_k}{P_k}$ for $k=1,2$.
  \end{itemize}

  Since $T_1$ and $T_2$ are choices, $S_1$ and $S_2$ must be
  branches. Since the combination of branches is only defined when
  they are exactly the same, we deduce
  $S_1 = S_2 = \dualof{\csum[r]{T_1}{T_2}}$.
  Using the induction hypothesis, we deduce that for every $k=1,2$
  there exist $\ContextB_{k1}$ and $\ContextB_{k2}$ such that
  $\csum[r]{\ContextB_{k1}}{\ContextB_{k2}} = \ContextB_k$ and
  $\wtp{\ContextB_{ki}, x : \dualof{T_i}}{P_k}$ for $i=1,2$.
  Let $\Context_i \eqdef \csum{\ContextB_{1i}}{\ContextB_{2i}}$ for
  $i=1,2$ and observe that
  $\csum[r]{\Context_1}{\Context_2} = \Context$.
  We conclude $\wtp{\Context_i, \x : \dualof{T_i}}{P}$ for $i=1,2$
  using \rulename{t-choice}.
\end{proof}

We now have all the ingredients to show that typing is preserved by
structural pre-congruence.

\begin{lemma}
  \label{lem:precong}
  If $\wtp\Context\P$ and $\P \spcong \Q$, then $\wtp\Context\Q$.
\end{lemma}
\begin{proof}
  By induction on the derivation of $P \spcong Q$ and by cases on
  the last rule applied. We only discuss a few selected cases, the
  others being simpler or trivial.

  \proofrule{s-no-choice}
  Then we have $\P = \csum[1]\Q\R$. From \rulename{t-choice} we
  deduce that there exist $\Context_1$ and $\Context_2$ such that
  $\Context = \csum[1]{\Context_1}{\Context_2}$ and
  $\wtp{\Context_1}\Q$ and $\wtp{\Context_2}\R$.
  Using Lemma~\ref{lem:no-choice} we conclude
  $\Context = \Context_1$.

  \proofrule{s-choice-idem}
  Then we have $\P = \csum\Q\Q$. From \rulename{t-choice} we deduce
  that there exist $\Context_1$ and $\Context_2$ such that
  $\Context = \csum{\Context_1}{\Context_2}$ and
  $\wtp{\Context_i}\Q$ for $i=1,2$.
  By Lemma~\ref{lem:choice-idem} we deduce
  $\Context_1 \probeq \Context_2$.
  It is a simple exercise to show that
  $\Context = \csum{\Context_1}{\Context_2}$ and
  $\Context_1 \probeq \Context_2$ imply
  $\Context = \Context_1 = \Context_2$, which suffices to conclude.

  \proofrule{s-par-choice}
  Then we have:
  \begin{itemize}
  \item $\P = (\csum{\P_1}{\P_2}) \ppar \R$;
  \item $\Q = \csum{(\P_1 \ppar \R)}{(\P_2 \ppar \R)$}.
  \end{itemize}

  From \rulename{t-par} and \rulename{t-choice} we deduce:
  \begin{itemize}
  \item
    $\Context = (\csum{\Context_1}{\Context_2}), \ContextB, x :
    \tsession{\pr{\csum{T_1}{T_2}}}$;
  \item $\wtp{\Context_i, x : T_i}{\P_i}$ for $i=1,2$;
  \item $\wtp{\ContextB, x : \dualof{\csum{T_1}{T_2}}}\R$.
  \end{itemize}

  Using Lemma~\ref{lem:splitting} we deduce that there exist
  $\ContextB_1$ and $\ContextB_2$ such that
  $\csum{\ContextB_1}{\ContextB_2} = \ContextB$ and
  $\wtp{\ContextB_i, x : \dualof{T_i}}\R$ for every $i=1,2$.
  We derive
  $\wtp{\ContextA_i, \ContextB_i, x : \tsession{\pr{T_i}}}{\P_i
    \ppar \R}$ for $i=1,2$ using \rulename{t-par}.
  We conclude
  $\wtp{(\csum{\ContextA_1}{\ContextA_2}),
    (\csum{\ContextB_1}{\ContextB_2}), x :
    \csum{\tsession{\pr{T_1}}}{\tsession{\pr{T_2}}}}\Q$ observing
  that
  \[
    \begin{array}{r@{~}c@{~}ll}
      \csum{\tsession{\pr{T_1}}}{\tsession{\pr{T_2}}} & = &
      \tsession{p\pr{T_1} + (1-p)\pr{T_2}} & \text{by Definition~\ref{def:ccomb}}
      \\
      & = & \tsession{\pr{\csum{T_1}{T_2}}} & \text{by Proposition~\ref{prop:ccomb}}
    \end{array}
  \]

  \proofrule{s-par-assoc}
  Then we have $P = (P_1 \ppar P_2) \ppar P_3$ and
  $Q = P_1 \ppar (P_2 \ppar P_3)$ and
  $\fn{P_2} \cap \fn{P_3} \ne \emptyset$.

  From \rulename{t-par} we deduce:
  \begin{itemize}
  \item $\Context = \ContextB, \Context_3, x : \tsession{\pr\T}$;
  \item $\wtp{\ContextB, x : T}{P_1 \ppar P_2}$;
  \item $\wtp{\Context_3, x : \dualof\T}{P_3}$.
  \end{itemize}

  From $\fn{P_2} \cap \fn{P_3} \ne \emptyset$ and
  $\dom\ContextB \cap \Context_3 = \emptyset$ we deduce
  $x\in\fn{P_2}$. Hence, from \rulename{t-par} we deduce:
  \begin{itemize}
  \item $\ContextB = \Context_1, \Context_2, y : \tsession{\pr\S}$;
  \item $\wtp{\Context_1, y : S}{P_1}$;
  \item $\wtp{\Context_2, x : T, y : \dualof\S}{P_2}$.
  \end{itemize}

  We derive
  $\wtp{\Context_2, \Context_3, x : T, y : \tsession{\pr\S}}{P_2
    \ppar P_3}$ with one application of \rulename{t-par} and we
  conclude
  $\wtp{\Context_1, \Context_2, \Context_3, x : \tsession{\pr\T}, y
    : \tsession{\pr\S}}{P_1 \ppar (P_2 \ppar P_3)}$ with another
  application of \rulename{t-par}.
\end{proof}

\restatesr*
\begin{proof}
  By induction on the derivation of $\P \lred \Q$ and by cases
  on the last rule applied. Since typing is syntax directed, in each
  case we can use the typing rule corresponding to the shape of the
  term under consideration.

  \proofrule{r-com}
  Then there exist $\x$, $\y$, $P_1$ and $P_2$ such that:
  \begin{itemize}
  \item $\P = \pout\x\y\P_1 \ppar \pin\x\y\P_2$;
  \item $\Q = \P_1 \ppar \P_2$.
  \end{itemize}

  From \rulename{t-par}, \rulename{t-out} and \rulename{t-in} we
  deduce that there exist $\Context_1$, $\Context_2$, $\t$ and $\T$
  such that:
  \begin{itemize}
  \item
    $\Context = \Context_1, \Context_2, \x :
    \tsession{\pr{\cout\t\T}}, \y : \t$;
  \item $\wtp{\Context_1, \x : \T}{\P_1}$;
  \item $\wtp{\Context_2, \x : \dualof\T, \y : \t}{\P_2}$.
  \end{itemize}

  We conclude $\wtp\Context{\P_1 \ppar \P_2}$ with one application
  of \rulename{t-par} and observing that $\pr{\cout\t\T} = \pr\T$ by
  Definition~\ref{def:pr}.

  \proofrule{r-left}
  Then there exist $x$, $P_1$, $Q_1$ and $Q_2$ such that:
  \begin{itemize}
  \item $\P = \pleft\x\P_1 \ppar \pbranch\x{\Q_1}{\Q_2}$;
  \item $\Q = \P_1 \ppar \Q_1$.
  \end{itemize}

  From \rulename{t-par}, \rulename{t-left} and \rulename{t-branch}
  we deduce that there exist $\Context_1$, $\ContextB$, $\T$ and
  $\S$ such that:
  \begin{itemize}
  \item
    $\Context = \Context_1, \ContextB, \x :
    \tsession{\pr{\tchoice[1]\T\S}}$;
  \item $\wtp{\Context_1, \x : \T}{\P_1}$;
  \item $\wtp{\ContextB, \x : \dualof\T}{\Q_1}$.
  \end{itemize}

  We conclude $\wtp\Context\Q$ with one application of
  \rulename{t-par} and observing that $\pr{\tchoice[1]\T\S} = \pr\T$
  by Definition~\ref{def:pr}.

  \proofrule{r-par}
  Then there exist $P_1$, $P_1'$ and $P_2$ such that:
  \begin{itemize}
  \item $\P = \P_1 \ppar \P_2$ for some $\P_1$ and $\P_2$;
  \item $\P_1 \lred \P_1'$;
  \item $\Q = \P_1'  \ppar \P_2$.
  \end{itemize}

  From \rulename{t-par} we deduce that there exist $\Context_1$,
  $\Context_2$, $\x$ and $\T$ such that:
  \begin{itemize}
  \item $\Context = \Context_1, \Context_2, x : \tsession{\pr\T}$;
  \item $\wtp{\Context_1, x : T}{\P_1}$;
  \item $\wtp{\Context_2, x : \dualof\T}{P_2}$.
  \end{itemize}

  Using the induction hypothesis we deduce
  $\wtp{\Context_1, x : T}{\P_1'}$. We conclude $\wtp\Context\Q$
  with one application of \rulename{t-par}.

  \proofrule{r-new}
  Then there exist $\x$, $\R$ and $\R'$ such that:
  \begin{itemize}
  \item $\P = \pnew\x\R$;
  \item $\R \lred \R'$;
  \item $\Q = \pnew\x\R'$.
  \end{itemize}

  From \rulename{t-new} we deduce that there exist $\ContextB$ and
  $\p$ such that:
  \begin{itemize}
  \item $\Context = \ContextB, x : \tsession\p$;
  \item $\wtp{\ContextB, \x : \tsession\p}\R$.
  \end{itemize}

  Using the induction hypothesis we deduce that
  $\wtp{\ContextB, \x : \tsession\p}{\R'}$ and we conclude
  $\wtp\Context\Q$ with one application of \rulename{t-new}.

  \proofrule{r-choice}
  Then there exist $P_1$, $P_1'$, $P_2$ and $p$ such that:
  \begin{itemize}
  \item $\P = \csum[\p]{\P_1}{\P_2}$;
  \item $\P_1 \lred \P_1'$;
  \item $\Q = \csum[\p]{\P_1'}{\P_2}$.
  \end{itemize}

  From \rulename{t-choice} we deduce that there exist $\Context_1$
  and $\Context_2$ such that:
  \begin{itemize}
  \item $\Context = \csum[\p]{\Context_1}{\Context_2}$;
  \item $\wtp{\Context_i}{\P_i}$ for all $i=1,2$.
  \end{itemize}

  Using the induction hypothesis we deduce $\wtp{\Context_1}{\P_1'}$
  and we conclude with one application of \rulename{t-choice}.

  \proofrule{r-struct}
  Then we have $\P \spcong \R \lred \R' \spcong \Q$ for some $\R$
  and $\R'$.
  From $\wtp\Context\P$ and Lemma~\ref{lem:precong} we deduce
  $\wtp\Context\R$.
  Using the induction hypothesis we deduce that $\wtp\Context{\R'}$.
  From Lemma~\ref{lem:precong} we conclude $\wtp\Context\Q$.
\end{proof}

%%% Local Variables:
%%% mode: latex
%%% TeX-master: "main"
%%% End:

\section{Proof of Theorem~\ref{thm:df}}
\label{sec:proofs_df}

In this appendix we develop the proof that well-typed processes are
deadlock free. First of all, we introduce the auxiliary notion of
\emph{hyper-context} which will be useful in the proof of
\cref{thm:df}.  An hypercontext $\HyperContext$ is a non-empty
multiset of contexts written
$\Context_1 \cpar \dots \cpar \Context_n$. We write
$\dom\HyperContext$ for the union of the domains of the contexts in
$\HyperContext$ and $\HyperContext \cpar \HyperContext'$ for the
multiset union of $\HyperContext$ and $\HyperContext'$.

If we think of a context as of the abstraction of well-typed
process, then an hyper-context intuitively represents a parallel
composition of such processes and a well-formed hyper-context is one
that represents a \emph{well-typed parallel composition} of the same
processes. Formally:

\begin{definition}[well-formed hyper-context]
  \label{def:wf}
  We say that $\HyperContext$ is \emph{well formed} if there exists
  $\Context$ such that $\wtc\HyperContext\Context$ is derivable
  using the following axiom and rule:
  \[
    \wtc\Context\Context
    \qquad
    \inferrule{
      \wtc\HyperContext{\ContextA, x : T}
      \\
      \wtc{\HyperContext'}{\ContextB, x : \dualof\T}
    }{
      \wtc{\HyperContext \cpar \HyperContext'}{\ContextA, \ContextB, x : \tsession{\pr\T}}
    }
  \]
\end{definition}

Note that the rightmost rule establishing the well formedness of an
hyper-context corresponds to \rulename{t-par} in the typing of
processes. A simple induction on the derivation of
$\wtc\HyperContext\Context$ suffices to establish that
$\dom\HyperContext = \dom\Context$.

We now show that there is a relationship between well-formed
hyper-contexts and the absence of cycles in the (contexts of the)
processes that are composed in parallel.

\begin{definition}[acyclic hyper-context]
  \label{def:cycle}
  We say that $\HyperContext$ has a \emph{cycle} if there exist $n$
  pairwise distinct $x_1, \dots, x_n$ and $n$ pairwise distinct
  $\Context_1, \dots, \Context_n \in \HyperContext$ with $n\geq2$
  such that
  $x_i \in \dom{\Context_i} \cap \dom{\Context_{(i \bmod n) + 1}}$
  for very $1\leq i\leq n$.
  We say that $\HyperContext$ is \emph{acyclic} if it has no cycle.
\end{definition}

\begin{proposition}
  \label{prop:wf-acyclic}
  If $\HyperContext$ is well formed, then it is acyclic.
\end{proposition}
\begin{proof}
  We prove a more general result, namely that
  $\wtc\HyperContext\Context$ implies that $\HyperContext$ is
  acyclic. We proceed by induction on the derivation of
  $\wtc\HyperContext\Context$.
  In the base case we have $\HyperContext = \Context$, hence
  $\HyperContext$ is acyclic because a cycle requires two or more
  contexts.
  Suppose $\HyperContext = \HyperContext_1 \cpar \HyperContext_2$
  and $\wtc{\HyperContext_1}{\Context_1, x : T}$ and
  $\wtc{\HyperContext_2}{\Context_2, x : \dualof\T}$ and
  $\Context = \Context_1, \Context_2, x : \tsession{\pr\T}$.
  By induction hypothesis both $\HyperContext_1$ and
  $\HyperContext_2$ are acyclic. Hence, any cycle of $\HyperContext$
  must involve two distinct names $x_1 \in \dom{\Context_1, x : T}$
  and $x_2 \in \dom{\Context_2, x : \dualof\T}$ that connect a
  context in $\HyperContext_1$ and a context in
  $\HyperContext_2$. However, $\HyperContext_1$ and
  $\HyperContext_2$ share just the name $\x$ because
  $\dom{\Context_1} \cap \dom{\Context_2} = \emptyset$. Therefore,
  $\HyperContext$ is acyclic.
\end{proof}

The next step towards the proof of deadlock freedom is to prove a
\emph{proximity lemma} showing that, whenever two well-typed
processes share a name -- that is, when they are connected by a
session -- it is always possible to rearrange them using structural
pre-congruence and respecting typing in such a way that they sit
next to each other and can possibly reduce. To do so, we introduce
some standard notation for process contexts:

\begin{definition}[process context]
  A \emph{process context} is a process containing a finite number
  of unguarded ``holes'' $\hole$. Formally, it is a term generated
  by the following grammar:
  \[
    \C, \D ~::=~ \hole ~\mid~ \P ~\mid~ \C \ppar \D ~\mid~ \csum\C\D ~\mid~ \pnew\x\C
  \]

  If $\C$ is a context with $n$ holes numbered from left to right
  according to the syntax of $\C$, we write $\C[P_1]\cdots[P_n]$ for
  the process obtained by filling the $i$-th hole with $P_i$. Note
  that filling a hole differs from substitution in that it may
  capture names, for example if $P_i$ is inserted in the scope of a
  binder. By writing $\C[P_1]\cdots[P_n]$, we implicitly assume that
  $\C$ has $n$ holes.
\end{definition}

Here is the proximity lemma. The hypothesis
$x \in (\fn\P \setminus \bn\C) \cap \fn\Q$ makes sure that the name
$x$ showing up in the context $\ContextA, x : T$ is the very same
$x$ that occurs free in $P$.

\begin{lemma}[proximity lemma]
  \label{lem:proximity}
  If $\wtp{\ContextA, x : T}{\C[P]}$ and
  $\wtp{\ContextB, x : \dualof\T}\Q$ and
  $x \in (\fn\P \setminus \bn\C) \cap \fn\Q$ and
  $\dom\ContextA \cap \dom\ContextB = \emptyset$, then there exists
  $\D$ such that $\C[P] \ppar Q \spcong \D[P \ppar Q]$ and
  $\wtp{\ContextA, \ContextB, x : \tsession{\pr\T}}{\D[P \ppar Q]}$.
\end{lemma}
\begin{proof}
  By induction on $\C$. We omit symmetric cases.

  \proofcase{$\C = \hole$}
  We conclude by taking $\D \eqdef \hole$ with one application of
  \rulename{t-par}.

  \proofcase{$\C = R \ppar \C'$}
  From \rulename{t-par} we deduce
  $\Context = \Context_1, \Context_2, y : \tsession{\pr\S}$ and
  $\wtp{\Context_1, y : S}\R$ and
  $\wtp{\Context_2, y : \dualof\S, x : T}{\C'[P]}$.
  Note that $x \ne y$, because the type of $x$ in the context used
  for typing $\C[P]$ is a session type and not a type of the form
  $\tsession\r$.
  Using the induction hypothesis we deduce that there exists $\D'$
  such that $\C'[P] \ppar Q \spcong \D'[P \ppar Q]$ and
  $\wtp{\Context_2, y : \dualof\S, \ContextB, x :
    \tsession{\pr\T}}{\D'[P \ppar Q]}$.
  Let $\D \eqdef R \ppar \D'$. We derive
  \[
    \begin{array}{@{}r@{~}c@{~}ll@{}}
      \C[P] \ppar Q & = & (R \ppar \C'[P]) \ppar Q & \text{by definition of $\C$}
      \\
      & \spcong & R \ppar (\C'[P] \ppar Q) & \text{by \rulename{s-par-assoc} using $\x\in\fn{\C'[P]} \cap \fn\Q$}
      \\
      & \spcong & R \ppar \D'[P \ppar Q] & \text{by property of $\D'$}
      \\
      & = & \D[P \ppar Q] & \text{by definition of $\D$}
    \end{array}
  \]
  and we conclude with one application of \rulename{t-par}.

  \proofcase{$\C = \csum\R{\C'}$}
  From \rulename{t-choice} we deduce
  $\Context, x : T = \csum{(\Context_1, x : T_1)}{(\Context_2, x :
    T_2)}$ and $\wtp{\Context_1, x : T_1}\R$ and
  $\wtp{\Context_2, x : T_2}{\C'[P]}$. In particular,
  $\dualof\T = \dualof{\csum{T_1}{T_2}}$.
  By Lemma~\ref{lem:splitting} we deduce that there exist
  $\ContextB_1$ and $\ContextB_2$ such that
  $\ContextB = \csum{\ContextB_1}{\ContextB_2}$ and
  $\wtp{\ContextB_i, \x : \dualof{\T_i}}\Q$ for $i=1,2$.
  Using the induction hypothesis we deduce that there exists $\D'$
  such that $\C'[P] \ppar Q \spcong \D'[P \ppar Q]$ and
  $\wtp{\ContextA_2, \ContextB_2, \x : \tsession{\pr{T_2}}}{\D'[P
    \ppar Q]}$.
  Let $\D \eqdef \csum{(\R \ppar \Q)}{\D'}$. We derive
  \[
    \begin{array}{@{}r@{~}c@{~}ll@{}}
      \C[P] \ppar Q & = & (\csum\R{\C'[P]}) \ppar Q & \text{by definition of $\C$}
      \\
      & \spcong & \csum{(R \ppar Q)}{(\C'[P] \ppar Q)} & \text{by \rulename{s-par-choice}}
      \\
      & \spcong & \csum{(R \ppar Q)}{\D'[P \ppar Q]} & \text{by property of $\D'$}
      \\
      & = & \D[P \ppar Q] & \text{by definition of $\D$}
    \end{array}
  \]

  We derive
  $\wtp{\ContextA_1, \ContextB_1, \x : \tsession{\pr{T_1}}}{R \ppar
    Q}$ using \rulename{t-par} and we conclude with one application
  of \rulename{t-choice}, observing that
  $\tsession{\pr\T} = \tsession{\pr{\csum{T_1}{T_2}}} =
  \csum{\tsession{\pr{T_1}}}{\tsession{\pr{T_2}}}$ by
  Proposition~\ref{prop:ccomb}.

  \proofcase{$\C = \pnew\y\C'$}
  From \rulename{t-new} we deduce
  $\wtp{\Context, y : \tsession{\pr\S}, x : T}{\C'[P]}$. Since $y$
  is bound we may assume, without loss of generality, that
  $y \not\in \fn\Q$.
  Using the induction hypothesis we deduce that there exists $\D'$
  such that $\C'[P] \ppar \Q \spcong \D'[P \ppar Q]$ and
  $\wtp{\Context, y : \tsession{\pr\S}, x : \tsession{\pr\T}}{\D'[P
    \ppar Q]}$.
  Let $\D \eqdef \pnew\y\D'$. We derive
  \[
    \begin{array}{@{}r@{~}c@{~}ll@{}}
      \C[P] \ppar Q & = & \pnew\y\C'[P] \ppar Q & \text{by definition of $\C$}
      \\
      & \spcong & \pnew\y(\C'[P] \ppar Q) & \text{by \rulename{s-par-new}}
      \\
      & \spcong & \pnew\y\D'[P \ppar Q] & \text{by property of $\D'$}
      \\
      & = & \D[P \ppar Q] & \text{by definition of $\D$}
    \end{array}
  \]
  and we conclude with one application of \rulename{t-new}.
\end{proof}

We now show that well-typed processes can be rewritten in a
\emph{normal form} in which all the restrictions and probabilistic
choices have been ``pushed outwards'', so that all the parallel
compositions concern sequential processes.

\begin{definition}[prefixed, sequential and exposed process]
  A process is \emph{prefixed} if it has the form $\pin\x\y\P$ or
  $\pout\x\y\P$ or $\pleft\x\P$ or $\pright\x\P$ or
  $\pbranch\x\P\Q$.
  A process is \emph{sequential} if it is either prefixed or it has
  the form $\pidle$ or $\pdone\x$ or $\pvar\A{\seqof\x}$.
  A process is \emph{exposed} if it is a parallel composition of
  sequential processes.
\end{definition}

\begin{definition}[process normal form]
  \label{def:nf}
  A process is in \emph{normal form} if it is generated by the
  grammar
  \[
    \Pnf ~::=~ \P ~\mid~ \pnew\x\Pnf ~\mid~ \csum[\p]\Pnf\Qnf
  \]
  where $\P$ is an exposed process.
\end{definition}

\begin{lemma}
  \label{lem:exposed-nf}
  If $P_1$ is in normal form and $P_2$ is exposed and
  $\wtp{\Context_1, x : T}{P_1}$ and
  $\wtp{\Context_2, x : \dualof\T}{P_2}$ and
  $\dom{\Context_1} \cap \dom{\Context_2} = \emptyset$, then there
  exists $P$ in normal form such that $P_1 \ppar P_2 \spcong P$ and
  $\wtp{\Context_1, \Context_2, x : \tsession{\pr\T}}\P$.
\end{lemma}
\begin{proof}
  A simple induction on the structure of $P_1$ recalling that it is
  in normal form. In the base case, when $P_1$ is exposed,
  $P_1 \ppar P_2$ is already in normal form and the result follows
  by reflexivity of $\spcong$ and one application of
  \rulename{t-par}. The inductive cases are analogous to the ones
  discussed in the proof of Lemma~\ref{lem:proximity}.
\end{proof}

\begin{lemma}
  \label{lem:par-nf}
  If $P_1$ and $P_2$ are in normal form and
  $\wtp{\Context_1, x : T}{P_1}$ and
  $\wtp{\Context_2, x : \dualof\T}{P_2}$ and
  $\dom{\Context_1} \cap \dom{\Context_2} = \emptyset$, then there
  exists $P$ in normal form such that $P_1 \ppar P_2 \spcong P$ and
  $\wtp{\Context_1, \Context_2, x : \tsession{\pr\T}}\P$.
\end{lemma}
\begin{proof}
  A simple induction on $P_2$ recalling that it is in normal
  form. In the base case, when $P_2$ is exposed, the result follows
  from Lemma~\ref{lem:exposed-nf}.
\end{proof}

\begin{lemma}[normal form]
  \label{lem:nf}
  If $\wtp\Context\P$, then there exists $\Q$ in normal form such
  that $\P \spcong \Q$ and $\wtp\Context\Q$.
\end{lemma}
\begin{proof}
  By induction on $\P$ and by cases on its shape.

  \proofcase{$P$ is sequential}
  Then $P$ is already in normal form and there is nothing left to prove.

  \proofcase{$P = \csum{P_1}{P_2}$}
  From \rulename{t-choice} we deduce
  $\Context = \csum{\Context_1}{\Context_2}$ and
  $\wtp{\Context_i}{P_i}$ for $i=1,2$.
  Using the induction hypothesis we deduce that there exist $Q_1$
  and $Q_2$ in normal form such that $P_i \spcong Q_i$ and
  $\wtp{\Context_i}{Q_i}$ for $i=1,2$.
  Let $Q \eqdef \csum{Q_1}{Q_2}$ and observe that $Q$ is in normal
  form. Now $P = \csum{P_1}{P_2} \spcong \csum{Q_1}{Q_2} = Q$ and we
  conclude $\wtp\Context\Q$ with one application of
  \rulename{t-choice}.

  \proofcase{$P = \pnew\x\P'$}
  From \rulename{t-new} we deduce
  $\wtp{\Context, x : \tsession\p}{\P'}$ for some $\p$.
  Using the induction hypothesis we deduce that there exists $\Q'$
  in normal form such that $P' \spcong Q'$ and
  $\wtp{\Context, x : \tsession\p}{Q'}$.
  Let $Q \eqdef \pnew\x\Q'$ and observe that $Q$ is in normal
  form. Now $P = \pnew\x\P' \spcong \pnew\x\Q' = Q$ and we conclude
  $\wtp\Context\Q$ with one application of \rulename{t-new}.

  \proofcase{$P = P_1 \ppar P_2$}
  From \rulename{t-par} we deduce
  $\Context = \Context_1, \Context_2, x : \tsession{\pr\T}$ and
  $\wtp{\Context_1, x : T}{P_1}$ and
  $\wtp{\Context_2, x : \dualof\T}{P_2}$.
  Using the induction hypothesis we deduce that there exist $Q_1$
  and $Q_2$ in normal form such that $P_i \spcong Q_i$ for $i=1,2$
  and $\wtp{\Context_1, x : T}{Q_1}$ and
  $\wtp{\Context_2, x : \dualof\T}{Q_2}$.
  We conclude using Lemma~\ref{lem:par-nf}.
\end{proof}

We now have almost all the ingredients for proving
\cref{thm:df}. The only aspect we have to consider is that the proof
will be an induction on the structure of the typing derivation,
hence the property that the process is well typed in the empty
context is not general enough to apply the induction hypothesis. We
generalize \cref{thm:df} by considering processes that are well
typed in \emph{balanced} contexts, assuring us that all the session
endpoints are used.

\begin{definition}[balanced type]
  \label{def:bal}
  We say that $\t$ is \emph{balanced} and we write $\bal\t$ if
  either $\un\t$ or $\t$ has the form $\tsession\p$ for some
  $\p$. We write $\bal\Context$ if $\bal{\Context(\x)}$ for every
  $\x\in\dom\Context$.
\end{definition}

\begin{lemma}
  \label{lem:df}
  If $\bal\Context$ and $\wtp\Context\P$ and $\P \nred$, then
  $\terminated\P$.
\end{lemma}
\begin{proof}
  Without loss of generality, we may assume that $\P$ is an exposed
  process. Indeed:
  \begin{itemize}
  \item If $\P$ is not in normal form, then Lemma~\ref{lem:nf}
    allows us to rewrite $\P$ into a normal form process that is
    well typed in the same $\Context$.
  \item If $\P$ is in normal form but not exposed, then it consists
    of top-level session restrictions and process distributions
    containing exposed processes, each of which is well typed in a
    balanced context and none of which reduces.
  \end{itemize}

  From the hypothesis $\P\nred$ we deduce that none of the
  sequential processes in $\P$ is a process invocation. Therefore,
  $\P$ is a parallel composition of
  $\P_1,\dots, \P_n, \Q_1,\dots, \Q_m$ where the $\P_i$ are prefixed
  processes and the $\Q_j$ are either $\pidle$ or of the form
  $\pdone\x$.
  From $\wtp\Context\P$ and \rulename{t-par} we deduce that there
  exist
  $\Context_1, \dots, \Context_n, \ContextB_1, \dots, \ContextB_m$
  such that $\wtp{\Context_i}{\P_i}$ for every $1\leq i\leq n$ and
  $\wtp{\ContextB_j}{\Q_j}$ for every $1\leq j\leq m$.
  Also, we let $\x_i$ be the channel that occurs in the prefix of
  $\P_i$. Clearly, $x_i \in \dom{\Context_i}$.
  We proceed by contradiction, assuming that $n \ne 0$.
  It must be the case that the $\x_i$ are pairwise distinct. Indeed,
  if $x_i = x_j$, then $x_i$ and $x_j$ would be the two peer
  endpoints of the same session performing complementary actions, by
  Lemma~\ref{lem:proximity} we would be able to move the two
  processes using $\x_i$ and $x_j$ next to each other and $\P$ would
  be able to reduce, thus contradicting the hypothesis
  $\P \nred$.
  Also, from the derivation of $\wtp\Context\P$ we can build a
  derivation of
  $\wtc{\Context_1 \cpar \dots \cpar \Context_n \cpar \ContextB_1
    \cpar \dots \ContextB_m}\Context$ according to
  Definition~\ref{def:wf}.

  The sub-structural nature of the type system and the hypothesis
  $\bal\Context$ ensure that each session name occurs exactly
  twice. Therefore, each $\x_i$ must also occur free in some other
  $\P_j$ with $j\ne i$.
  We let $f : [1,n] \to [1,n]$ be the function that maps $i$ to the
  index of the process in which $x_i$ occurs free. That is,
  $\x_i \in \fn{\P_{f(i)}}$ for every $1\leq i\leq n$.
  Note that $f(i) \ne i$ by definition of $f$. Now we build the
  following infinite sequence of names
  \[
    \x_1,
    \x_{f(1)},
    \x_{f(f(1))},
    \x_{f(f(f(1)))},
    \dots
  \]

  Since there are $n$ distinct names $x_i$ and $f(i) \ne i$, there
  are at least two names that occur infinitely often in this
  sequence. Consequently, the hyper-context
  $\Context_1 \cpar \dots \cpar \Context_n \cpar \ContextB_1 \cpar
  \dots \cpar \ContextB_m$ must have a cycle in the sense of
  Definition~\ref{def:cycle}, which contradicts
  Proposition~\ref{prop:wf-acyclic}.
\end{proof}

\restatedf*
\begin{proof}
  Immediate consequence of Theorem~\ref{thm:sr} and
  Lemma~\ref{lem:df}.
\end{proof}

%%% Local Variables:
%%% mode: latex
%%% TeX-master: "main"
%%% End:

\section{Proof of Theorem~\ref{thm:soundness}}
\label{sec:proofs_soundness}

\begin{lemma}
  \label{lem:sound}
  If $\wtp{\Context, x : \T}\P$ and $\terminated\P$, then
  $\P \success\x{\pr\T}$.
\end{lemma}
\begin{proof}
  By induction on the derivation of $\wtp{\Context, x : \T}\P$ and
  by cases on the last rule applied. We only discuss those cases
  that are compatible with the hypothesis $\terminated\P$.

  \proofrule{t-idle}
  Then $P = \pidle$ and $\un\T$, hence $\T = \tend$. We conclude
  $\P \success\x{0}$ noting that $\pr\T = 0$.

  \proofrule{t-done}
  Then $P = \pdone\y$. We distinguish two subcases. If $x = y$, then
  $T = \tdone$ and we conclude $\P \success\x{1}$ noting that
  $\pr\T = 1$. If $x \ne y$, then we have $\un\T$, hence
  $\T = \tend$ and we conclude as in the case of rule
  \rulename{t-idle}.

  \proofrule{t-par}
  Then $P = P_1 \ppar P_2$ and
  $\Context, x : T = \Context_1, \Context_2, y : \tsession{\pr\S}$
  and $\wtp{\Context_1, y : S}{P_1}$ and
  $\wtp{\Context_2, y : \dualof\S}{P_2}$.
  It must be the case that $x\in\dom{\Context_i}$ for some
  $i\in\set{1,2}$. We conclude using the induction hypothesis on
  $P_i$.

  \proofrule{t-choice}
  Then there exist $P_1$, $P_2$, $\Context_1$, $\Context_2$, $T_1$
  and $T_2$ such that $P = \csum{P_1}{P_2}$ and
  $\Context, x : T = \csum{\Context_1, x : T_1}{\Context_2, x :
    T_2}$ and $\wtp{\Context_i, x : T_i}{P_i}$ for $i=1,2$.
  From $\terminated\P$ we deduce $\terminated{P_i}$ for $i=1,2$.
  Using the induction hypothesis we deduce
  $P_i \success\x{\pr{T_i}}$ for $i=1,2$, hence
  $P \success\x{p\pr{T_1}+(1-p)\pr{T_2}}$ by
  Definition~\ref{def:success}.
  We conclude $P \success\x{\pr\T}$ using
  Proposition~\ref{prop:ccomb}.

  \proofrule{t-new}
  Then there exist $Q$, $y$ and $p$ such that $P = \pnew\y\Q$ and
  $\wtp{\Context, y : \tsession\p, x : T}\Q$.
  From $\terminated\P$ we deduce $\terminated\Q$.
  We conclude using the induction hypothesis.
\end{proof}

\restatesound*
\begin{proof}
  From Lemma~\ref{lem:df} we deduce $\terminated\P$.
  We prove that $\wtp{\Context, x : \tsession\p}\P$ and
  $\terminated\P$ imply $\P \success\x\p$ by induction on the
  derivation of $\wtp{\Context, x : \tsession\p}\P$ and by cases on
  the last rule applied. We only consider those cases that are
  compatible with the assumption $x : \tsession\p$.

  \proofcase{\rulename{t-par} when the name being split is $x$}
  Then there exist $P_1$, $P_2$, $\Context_1$, $\Context_2$ and $T$
  such that $P = P_1 \ppar P_2$ and
  $\Context = \Context_1, \Context_2$ and
  $\wtp{\Context_1, x : T}{P_1}$ and
  $\wtp{\Context_2, x : \dualof\T}{P_2}$ and $p = \pr\T$.
  From $\terminated\P$ we deduce $\terminated{P_1}$.
  We conclude using Lemma~\ref{lem:sound}.

  \proofcase{\rulename{t-par} when the name being split is some $y \ne x$}
  Then there exist $P_1$, $P_2$, $\Context_1$, $\Context_2$ and $T$
  such that $P = P_1 \ppar P_2$ and
  $\Context, x : \tsession\p = \Context_1, \Context_2, y :
  \tsession{\pr\T}$ and $\wtp{\Context_1, y : T}{P_1}$ and
  $\wtp{\Context_2, y : \dualof\T}{P_2}$.
  We only discuss the case $x\in\dom{\Context_1}$, the other being
  analogous. Then $\Context_1 = \Context_1', x : \tsession\p$ for
  some $\Context_1'$.
  From $\terminated\P$ we deduce $\terminated{P_1}$. We conclude
  using the induction hypothesis.

  \proofrule{t-choice}
  Then there exist $P_1$, $P_2$, $\Context_1$, $\Context_2$, $q$,
  $p_1$ and $p_2$ such that $P = \csum[q]{P_1}{P_2}$ and
  $\Context, x : \tsession\p = \csum[q]{(\Context_1, x :
    \tsession{p_1})}{(\Context_2, x : \tsession{p_2})}$ and
  $\wtp{\Context_i, x : \tsession{p_i}}{P_i}$ for $i=1,2$. In
  particular, $p = qp_1+(1-q)p_2$.
  From $\terminated\P$ we deduce $\terminated{P_i}$ for $i=1,2$.
  Using the induction hypothesis we deduce $P_i \success\x{p_i}$ for
  $i=1,2$, hence we conclude $P \success\x{qp_1+(1-q)p_2}$.

  \proofrule{t-new}
  Then there exist $y$, $p$, $Q$ such that $P = \pnew\y\Q$ and
  $\wtp{\Context, y : \tsession\q, x : \tsession\p}\Q$.
  From $\terminated\P$ we deduce $\terminated\Q$. We conclude using
  the induction hypothesis.
\end{proof}

\restaterelative*
\begin{proof}
  We prove the two items separately.

  \proofcase{Item 1}
  From the hypothesis $\terminates\P\p$ we know that there exist
  $(Q_n)$, $(R_n)$ and $(p_n)$ such that
  $P \wred \csum[p_n]{Q_n}{R_n}$ and $\terminated{Q_n}$ for every
  $n\in\mathbb{N}$ and $\lim_{n\to\infty} p_n = p$.
  From the hypothesis $\wtp{x : \tsession{1}}\P$ and \cref{thm:sr}
  we deduce $\wtp{x : \tsession{1}}{\csum[p_n]{Q_n}{R_n}}$ for every
  $n\in\mathbb{N}$.
  From \rulename{t-choice} and \cref{def:ccomb} we deduce
  $\wtp{x : \tsession{1}}{Q_n}$ for every $n\in\mathbb{N}$.
  From $\terminated{Q_n}$ and \cref{thm:soundness} we deduce
  $Q_n \success\x{1}$.
  Using \cref{def:success} we derive
  $\csum[p_n]{Q_n}{R_n} \success\x{\p_n}$ for every
  $n\in\mathbb{N}$, hence $P \lsuccess\x\p$.

  \proofcase{Item 2}
  From the hypothesis $\terminates\P{1}$ we know that there exist
  $(Q_n)$, $(R_n)$ and $(p_n)$ such that
  $P \wred \csum[p_n]{Q_n}{R_n}$ and $\terminated{Q_n}$ for every
  $n\in\mathbb{N}$ and $\lim_{n\to\infty} p_n = p$.
  That is, for every $\varepsilon > 0$, there exists $N$ such that,
  for every $n\geq N$, we have $1 - p_n < \varepsilon$.
  From the hypothesis $\wtp{x : \tsession\p}\P$ and \cref{thm:sr} we
  deduce $\wtp{x : \tsession\p}{\csum[p_n]{Q_n}{R_n}}$ for every
  $n\in\mathbb{N}$.
  From \rulename{t-choice} and \cref{def:ccomb} we deduce that, for
  every $n\in\mathbb{N}$, there exist $q_n$ and $r_n$ such that
  $p = p_nq_n + (1-p_n)r_n$ and $\wtp{x : \tsession{q_n}}{Q_n}$ and
  $\wtp{x : \tsession{r_n}}{R_n}$.
  From $\terminated{Q_n}$ and \cref{thm:soundness} we deduce
  $Q_n \success\x{q_n}$ for every $n\in\mathbb{N}$, hence
  $(\csum[p_n]{Q_n}{R_n}) \success\x{p_nq_n}$ for every
  $n\in\mathbb{N}$.
  Now
  $p - p_nq_n = p_nq_n + (1 - p_n)r_n - p_nq_n = (1 - p_n)r_n <
  \varepsilon$, hence $\lim_{n\to\infty} p_nq_n = p$.
\end{proof}

%%% Local Variables:
%%% mode: latex
%%% TeX-master: "main"
%%% End:

\end{document}